\documentclass{llncs}
\usepackage[utf8]{inputenc}
\usepackage[T1]{fontenc}
\usepackage{bm}
\usepackage{amsmath,amsfonts,amssymb,multirow}
\usepackage{multicol}
\usepackage{fullpage}
\usepackage{epsfig}
\usepackage{graphics}
\usepackage{float}

\usepackage[dvipsnames,usenames]{color}
\usepackage[colorlinks=true,urlcolor=Blue,citecolor=Green,linkcolor=BrickRed]{hyperref}
\usepackage[usenames,dvipsnames]{xcolor}
\makeatletter

\let\markeverypar\everypar
\newtoks\everypar
\everypar\markeverypar
\markeverypar{\the\everypar\looseness=-1}

\newcommand{\polylog}{\text{\,polylog}}
\newcommand{\pred}[1]{pred(#1)}

\begin{document}

\title{Submatrix Maximum Queries in Monge Matrices\\ are
Equivalent to Predecessor Search\thanks{This paper is based on two preliminary papers that appeared in ICALP 2014 and ICALP 2015.}}

\author{
Pawe{\l} Gawrychowski\inst{1} \and Shay Mozes\inst{2}\thanks{Mozes and Weimann supported in part by
  Israel Science Foundation grant 794/13.} \and Oren Weimann\inst{3}$^{\star\star}$
}
\institute{
University of Wrocław, \href{mailto:gawry@cs.uni.wroc.pl}{gawry@cs.uni.wroc.pl} \and
IDC Herzliya, \href{mailto:smozes@idc.ac.il}{smozes@idc.ac.il}  \and
University of Haifa, \href{mailto:oren@cs.haifa.ac.il}{oren@cs.haifa.ac.il}  
}

\date{}
\maketitle

\begin{abstract}
We present an optimal data structure for submatrix maximum
queries in $n\times n$ Monge matrices.
Our result is a two-way reduction showing that the problem is
equivalent to the classical predecessor problem in a universe of polynomial size. This gives a data
structure of $O(n)$ space that answers submatrix maximum queries in
$O(\log\log n)$ time, as well as a matching lower bound, showing that $O(\log\log n)$ query-time is optimal for any data structure
of size $O(n\polylog(n))$. 
Our result settles the problem, improving on the $O(\log^2 n)$
query-time in SODA'12, and on the $O(\log n)$ query-time in ICALP'14. 

In addition, we show that partial Monge matrices can be handled in the same bounds as full Monge matrices. In both previous results, partial Monge matrices incurred additional inverse-Ackermann factors.
\end{abstract}

\section{Introduction}
Data structures for range queries and for predecessor queries are
among the most studied data structures in computer science. Given an
$n \times n$ matrix $M$,  a {\em range maximum} (also called submatrix
maximum) data structure can report the maximum entry in any query
submatrix (a set of consecutive rows and a set of consecutive columns)
of $M$. Given a set $S\subseteq [0,U)$ of $n$ integers from a
polynomial universe $U$, 
a {\em predecessor} data structure can report the predecessor (and successor) in $S$ of any query integer $x\in  [0,U)$. 
In this paper, we prove that these two seemingly unrelated problems are in fact equivalent when the matrix $M$ is a {\em Monge} matrix.

\paragraph{\bf Range maximum queries.}
A long line of research over the last three decades including~\cite{AFL07,CR1989,DemaineLandauWeimann,GBT84,YuanA10} achieved range maximum data structures of $\tilde O(n^2)$ space and $\tilde O(1)$ query time\footnote{The $\tilde O(\cdot)$ notation
  hides polylogarithmic factors in $n$.}, culminating with the $O(n^2)$-space  $O(1)$-query data structure of Yuan and Atallah~\cite{YuanA10}. In general matrices, this is optimal since representing the input matrix already
requires $\Theta(n^2)$ space. In fact, reducing the additional space to $O(n^2/c)$ is known to incur an $ \Omega(c)$ query-time~\cite{BrodalDR10} and such tradeoffs can indeed be achieved for any value of $c$~\cite{BrodalESA,BrodalDR10}. 

However, in many applications, the matrix $M$ is not stored explicitly
but any entry of $M$ can be computed when needed in $O(1)$ time. One
such case is when the matrix $M$ is sparse, i.e.,  has 
 $N = o(n^2)$ nonzero entries. In this case 
the problem is known in computational geometry as the {\em orthogonal range searching} problem  on the $n \times n$ grid. In this case as well, various data structures with $\tilde O(N)$-space and $\tilde O(1)$-query appear in a long history of results including
~\cite{AlstrupEtAl,Patrascu,Chazelle88,Munro,GBT84}. 
For a survey on orthogonal range searching see~\cite{Nekrich}. 
Another case where the additional space can be made $o(n^2)$ (and in fact even $O(n)$) is when the matrix is a Monge matrix.

\paragraph{\bf  Range maximum queries in Monge matrices.}
A matrix $M$ is {\em Monge} if for any pair of rows $i<j$ and columns
$k<\ell$ we have that $M[i,k]+ M[j,\ell] \ge M[i,\ell]+
M[j,k]$.\footnote{Monge matrices are often defined with a $\le$
  (rather than $\ge$) in the condition. Our results apply to both
  definitions, as well as to minimum (rather than maximum) queries.}
 A matrix $M$ is {\em Totally Monotone} (or TM) if for any pair of
 rows $i<j$ and columns $k<\ell$ we have that if $M[i,k]\le M[j,k]$
 then $M[i,\ell] \le M[j,\ell]$.  Notice that the Monge property
 implies total monotonicity but the converse is not true. Whenever
 possible, we state our results for the more general class of TM
 matrices. Throughout the paper we use a top-down and left-to-right ordering of the
 elements of a matrix. We say that $M[i,k]$ is above $M[j,\ell]$ if
 $i<j$, and to the left  of $M[j,\ell]$ if $k<\ell$.
 
Submatrix maximum queries on Monge matrices 
have various important applications in combinatorial optimization and computational
geometry 
such as problems involving distances in the plane, and in problems on
convex $n$-gons. 
See~\cite{BKR96} for a survey on Monge matrices and their uses in combinatorial
optimization. Submatrix maximum queries on Monge matrices are
used in algorithms that efficiently find the largest empty rectangle containing a
query point, in dynamic distance oracles for planar graphs,
and in algorithms for maximum flow in planar graphs.
See~\cite{KaplanMozesNussbaumSharir} for more details on the
history of this problem and its applications.

Given an $n \times n$ Monge matrix $M$ it is possible to obtain compact data structures of only $\tilde O(n)$ space that can answer submatrix maximum queries in $\tilde{O}(1)$ time.  The first such data structure was given by 
Kaplan, Mozes, Nussbaum and
Sharir~\cite{KaplanMozesNussbaumSharir}. 
They presented an $O(n\log n)$-space data structure with  $O(\log^2
n)$ query time.
This was improved in~\cite{ourICALP}  to $O(n)$ space and $O(\log n)$ query time.  

\paragraph{\bf Breakpoints and Partial Monge matrices.}
Given an $m\times n$ Monge matrix $M$, let $r(c)$ be the row containing the maximum element in the $c$-th
column of $M$. It is easy to verify that the $r(\cdot)$ values are monotone, i.e., $r(1)\leq r(2) \leq\ldots\leq r(n)$.
Columns $c$ such that $r(c-1)<r(c)$ (or $c=1$) are called the {\em breakpoints} of $M$. A Monge matrix
consisting of $m<n$ rows has $O(m)$ breakpoints, which can be found in
$O(n)$ time using the SMAWK algorithm~\cite{SMAWK} (total monotonicity
suffices for SMAWK). 

Some applications involve  {\em partial} Monge matrices rather than
full Monge matrices.
A partial matrix is a matrix where some of the
entries are undefined, but the defined entries in each row
and in each column are contiguous. A partial Monge matrix is a partial matrix in which the Monge inequality is satisfied whenever all four involved entries are defined. The total number of breakpoints in a partial Monge matrix is still $O(m)$ (as we show in Section~\ref{sec:partial}), and they can be found in $O(n\cdot \alpha(n))$ time\footnote{Here $\alpha(n)$ is the inverse-Ackermann function.} using an algorithm of Klawe and Kleitman~\cite{KK89}. This was used in~\cite{ourICALP,KaplanMozesNussbaumSharir} to extend their solutions to partial Monge matrices at the cost of an additional $\alpha(n)$ factor to the query time.\footnote{In~\cite{KaplanMozesNussbaumSharir}, there was also an additional $\log n$ factor to the space.}

\paragraph{\bf Our results.} 
In this paper, we fully resolve the submatrix maximum
query problem in $n\times n$ Monge matrices by presenting a data structure of $O(n)$ space and $O(\log\log n)$ query time. 
Consequently, we obtain an improved query time for other applications such as finding the largest empty
rectangle containing a query point. 
We compliment our upper bound with a matching lower bound, showing
that $O(\log\log n)$ query-time is optimal for any data structure of
size $O(n\polylog(n))$. 
Implicit in our upper and lower bound is an equivalence
between the predecessor problem in a universe of polynomial size and the range maximum query problem in
Monge matrices. The upper bound essentially reduces a submatrix
query to a constant number of predecessor problems, and vice versa, the lower bound
reduces the predecessor problem to a submatrix query problem. In fact, the lower bound holds even for the more restricted case where the submatrix query is a subcolumn.

Finally, we extend our result to partial Monge matrices with the exact
same bounds (i.e., $O(n)$ space and $O(\log\log n)$ query time). Our
result is the first to achieve such extension with no overhead.

\paragraph{\bf Techniques.} 
Let $M$ be an $n \times n$ Monge matrix\footnote{We consider $m \times n$ matrices, but for simplicity we sometimes state the
results for $n \times n$ matrices.}. Consider a full binary tree
$\mathcal T$ whose leaves are the rows of $M$. Let $M_u$ be the
submatrix of $M$ composed of all rows (i.e., leaves) in the subtree of
a node $u$ in $\mathcal T$.
Both existing data structures for submatrix maximum queries~\cite{ourICALP,KaplanMozesNussbaumSharir} store, for each node $u$ in $\mathcal T$ a data structure $D_u$. The goal of $D_u$ is to answer submatrix maximum queries  that include an arbitrary interval of columns and {\em exactly all rows}  of $M_u$.   
This way, an arbitrary query is covered
in~\cite{ourICALP,KaplanMozesNussbaumSharir} by querying the $D_u$ structures of $O(\log n)$ canonical nodes of $\mathcal T$.  An $\Omega(\log n)$ bound is thus inherent for any solution that examines the canonical nodes. 
We overcome this obstacle by designing a stronger data
structure $D_u$. Namely, one that supports  queries
that include an arbitrary interval of columns and {\em a prefix of
  rows} or {\em a suffix of rows} of $M_u$. 
This way, an arbitrary query can be covered by just two $D_u$s. The
idea behind the new design  is to efficiently  encode the changes in column maxima as we add rows to $M_u$ one by one. Retrieving this information is done using weighted ancestor search and range maximum queries on trees. This is a novel use of these techniques.     

For our lower bound, we show that for any set of $n$ integers
$S\subseteq [0,n^{2})$ there exists an $n\times n$ Monge matrix $M$ such that the predecessor of $x$ in $S$
can be found with submatrix maximum queries on $M$. 
The predecessor lower bound of
P{\v{a}}tra{\c{s}}cu and Thorup~\cite{PT2006} then implies that $O(n \polylog(n))$ space requires  $\Omega(\log\log n)$ query time. 
We overcome two technical
difficulties here: First, $M$ should be Monge. Second, there must be an  $O(n\polylog(n))$-size representation of $M$ which can retrieve any entry $M[i,j]$ in $O(1)$ time. 

Finally, for handling partial Monge matrices, and unlike previous
solutions for this case, we do not directly adapt the solution for the
full Monge case to partial Monge matrices. Instead we decompose the partial Monge matrix into
many full Monge matrices, that can be preprocessed to be queried
cumulatively in an efficient way. This requires significant technical work and
careful use of the structure of the decomposition.

\paragraph{\bf Computational model.}
We assume the standard word RAM model with word size $\Omega(\log n)$. However,
this is just an internal assumption and the elements of the matrix $M$ are only accessed
through a comparison oracle, that is, we only assume that we are able to check in constant
time if $M[i,j] \leq M[i',j']$ and no arithmetical manipulation on the elements of $M$ is 
performed.

\paragraph{\bf Roadmap.}
In Section~\ref{sec:structure} we present an $O(n\log n)$-space data structure for Monge matrices that answers submatrix maximum queries in $O(\log \log n)$ time. In Section~\ref{sec:linear} we reduce the space to $O(n)$. Our lower bound is given in Section~\ref{sec:lower bound}, and the extension to partial Monge matrices in Section~\ref{sec:partial}.

\section{Data structure for Monge matrices}
\label{sec:structure}
Our goal in this section is to construct, for a given $m\times n$  Monge matrix $M$, a data structure
of size $O(m\log n)$ that answers submatrix maximum queries in $O(\log\log n)$ time. In Section~\ref{sec:linear}
we show how to reduce the space from $O(n\log n)$ to $O(n)$ when $m=n$.
We will actually show a stronger result, namely the structure allows us to reduce in $O(1)$ time
a submatrix maximum query into $O(1)$ predecessor queries on a set consisting of $n$ integers from
a polynomial universe.

We denote by $\pred{m,n}$ the complexity of a predecessor query on a set of $m$ integers from
a universe $\{0,\ldots,n-1\}$. It is well known
that there are $O(m)$-space data structures achieving $\pred{m,n} = \min\{O(\log m),O(\log\log n)\}$.

Recall
that a submatrix maximum query returns the maximum $M[i,j]$ over all $i\in [i_{0},i_{1}]$ and
$j\in [j_{0},j_{1}]$ for given $i_{0}\leq i_{1}$ and $j_{0}\leq
j_{1}$. We start by answering the easier 
 \emph{subcolumn maximum queries} within these space and time bounds. 
That is,  finding the maximum $M[i,j]$ over all $i\in [i_{0},i_{1}]$ for given
$i_{0}\leq i_{1}$ and $j$.

We construct a full binary tree $\mathcal T$ over the rows of $M$. Every leaf of the tree corresponds to a single row
of $M$, and every inner node corresponds to the range of rows in its subtree. To find the maximum $M[i,j]$ over all
$i\in [i_{0},i_{1}]$ for  given $i_{0}\leq i_{1}$ and $j$, we first
locate the lowest common ancestor (lca) $u$ of the leaves
corresponding to $i_{0}$ and $i_{1}$ in the tree. Then we decompose the query into two parts:
one fully within the range of rows $M_{\ell}$ of the left child of $u$, and one fully within
the range of rows $M_{r}$ of the right child of $u$. The former ends at the last row of $M_{\ell}$
and the latter starts at the first row of $M_{r}$. We equip every node with two data structures
supporting  such simpler subcolumn maximum queries. Because of symmetry (if $M$ is
Monge, so is $M'$, where $M'[i,j]=M[n+1-i,n+1-j]$) it suffices to show how to answer
subcolumn maximum queries starting at the first row. 

\begin{lemma}
\label{lem:subcolumn}
Given an $m\times n$ TM matrix $M$, a data structure of size $O(m)$ can be constructed
in $O(m\log n)$ time to answer in $O(\pred{m,n})$ time subcolumn maximum queries starting at the first row of $M$.
\end{lemma}

\begin{proof}
Consider queries spanning an {\em entire} column $c$ of $M$. To answer such a query, we only need to find
the corresponding $r(c)$. If we store the breakpoints of $M$ in a predecessor structure, where
every breakpoint $c$ links to its corresponding value of $r(c)$, a query can be answered
with a single predecessor search. More precisely, to determine the maximum in the $c$-th
column of $M$, we locate the largest breakpoint $c' \leq c$, and then set $r(c)=r(c')$.
Hence we can construct a data structure of size $O(m)$ to answer {\em entire column} maximum
queries in $O(\pred{m,n})$ time.

Let $M_{i}$ be a TM matrix consisting of the first $i$ rows of $M$. By applying the above
reasoning to every $M_{i}$ separately, we immediately get a structure of size $O(m^{2})$ answering
subcolumn maximum queries starting at the first row of $M$ in $O(\pred{m,n})$ time. We
want to improve on this by utilizing the dependency of the structures constructed for different $i$'s. Observe that the list of breakpoints of $M_{i+1}$ is a prefix
of the list of breakpoints of $M_{i}$ to which we append at most one new element. In other words,
if the breakpoints of $M_i$ are stored on a stack, we need to pop zero or more elements and push
at most one new element to represent the breakpoints of $M_{i+1}$. Consequently, instead of storing a separate list for every $M_{i}$,
we can succinctly describe the content of all stacks with a single tree $T$ on at most $m+1$
nodes. For every $i$, we store a pointer to a node $s(i)\in T$, such that the ancestors of $s(i)$
(except for the root) are exactly the breakpoints of $M_{i}$.
Whenever we pop an element from the current stack, we move to the parent of the current
node, and whenever we push an element, we create a new node and make it a child of the
current node. Initially, the tree consists of just the root. Every node is labelled with a column
number and by construction these numbers are strictly increasing on any path starting at
the root (the root is labelled with $-\infty$). 
Therefore, a predecessor search for $j$ among the 
breakpoints of $M_{i}$ reduces to finding the leafmost ancestor of $s(i)$ whose label is at most $j$. 
This is known as the {\em weighted ancestor} problem.
Weighted ancestor queries on a tree of size $O(m)$
are equivalent to predecessor searching on a number of sets of $O(m)$
total size~\cite{Kopelowitz},\footnote{The reduction described in~\cite{Kopelowitz}
needs $O(\log^{*}m)$ additional time and (adaptively) queries two sets. The additional
time is required to reduce the total size of the sets to $O(m)$, which is done
by recursively decomposing the tree. However, this recursive decomposition
can be avoided using atomic heaps as explained in Lemma 11 of~\cite{GawrychowskiLN14}.
Then in $O(1)$ additional time we are able to reduce a weighted ancestor query
to a single predecessor query in one of the sets.}
achieving the claimed space and query time bounds.

To finish the proof, we need to bound the construction time. The bottleneck is constructing
the tree $T$. Let $c_{1}< c_{2}<\ldots <c_{k}$ for some $k\le i$ be the breakpoints of $M_{i}$. As long as
$M[i+1,c_{k}]\geq M[r(c_{k}),c_{k}]$ we decrease $k$ by one, i.e., remove the last breakpoint.
This process is repeated $O(m)$ times in total. If $k=0$ we create a new breakpoint
$c_{1}=1$. 
If $k\geq 1$ and  $M[i+1,c_{k}] <M[r(c_{k}),c_{k}]$, we check if $M[i+1,n] \geq M[r(c_{k}),n]$.
If so, we need to create a new breakpoint. To this end, we need to find the smallest $j$ such
that $M[i+1,j] \geq M[r(c_{k}),j]$. This can be done in
$O(\log n)$ using binary search. Consequently, $T$ can be constructed in $O(m\log n)$ time. Then augmenting
it with a weighted ancestor structure takes $O(m)$ time.
\qed \end{proof}

We apply Lemma~\ref{lem:subcolumn} twice to every node of the full version tree $\mathcal T$.
Once for subcolumn maximum queries starting at the first row and once for queries ending at the last row. Since the total
size of all structures at the same level of the tree is $O(m)$, the total size of our subcolumn
maximum data structure becomes $O(m\log m)$, and it can be constructed in $O(m\log m\log n)$
time to answer queries in $O(\pred{m,n})$ time. Hence we have proved
the following.

\begin{theorem}
\label{thm:subcolumn}
Given an $m\times n$ TM matrix $M$, a data structure of size $O(m\log m)$ can be constructed
in $O(m\log m\log n)$ time to answer subcolumn maximum queries in $O(\pred{m,n})$ time.
\end{theorem}

By symmetry (a transpose of a Monge matrix is Monge) we can answer subrow maximum
queries (where the query is a single row and a range of columns) in
$O(\pred{n,m})$ time. We are now ready to
tackle general submatrix maximum queries.

At a high level, the idea is identical to the one used for subcolumn maximum queries: we construct
a~full binary tree $\mathcal T$ over the rows of $M$, where every node corresponds to a range of rows. To
find maximum $M[i,j]$ over all $i\in [i_{0},i_{1}]$ and
$j\in [j_{0},j_{1}]$ for given $i_{0}\leq i_{1}$ and $j_{0}\leq j_{1}$, we locate the lowest
common ancestor of the leaves corresponding to $i_{0}$ and $i_{1}$ and decompose
the query into two parts, the former ending at the last row of $M_{\ell}$ and the latter
starting at the first row of $M_{r}$. Every node is equipped with two data structures
allowing us to answer submatrix maximum queries starting at the first row or ending
at the last row. As before, it suffices to show how to answer submatrix maximum queries starting
at the first row.

\begin{lemma}
\label{lem:submatrix}
Given an $m\times n$ Monge matrix $M$, and a data structure that
answers subrow maximum queries on $M$ in $O(\pred{n,m})$ time, one can
construct in $O(m\log m)$ time  a data structure consuming $O(m)$
additional space, that answers submatrix maximum queries starting at
the first row of $M$ in $O(\pred{m,n}+\pred{n,m})$ time.
\end{lemma}

\begin{proof}
We extend the proof of Lemma~\ref{lem:subcolumn}.  Let $c_{1} < c_{2} < \ldots <c_{k}$
be the breakpoints of $M$ stored in a predecessor structure. For every $i \geq 2$
we precompute and store the value  $$m_{i}=\max_{j\in [c_{i-1},c_{i})} M[r(c_{i-1}),j].$$
These values are
augmented with a (one dimensional) range maximum query data structure. 
To begin with, consider a submatrix maximum query starting at the first row of $M$ and ending
at the last row of $M$, i.e., we need to calculate the maximum $M[i,j]$ over all $i\in [1,m]$ and $j\in [j_{0},j_{1}]$. 
We find in $O(\pred{m,n})$ the successor of $j_{0}$, denoted $c_{i}$, and the predecessor of $j_{1}$, denoted
$c_{i'}$. There are  three possibilities:
\begin{enumerate}
\item The maximum is reached for $j\in [j_{0},c_{i})$,
\item The maximum is reached for $j\in [c_{i},c_{i'})$,
\item The maximum is reached for $j\in [c_{i'},j_{1})$.
\end{enumerate}
The first and the third possibilities can be calculated with subrow maximum queries in $O(\pred{n,m})$ time, because
both ranges span an interval of columns and a single row.
The second possibility can be calculated with a range maximum query on the range $(i,i']$ over
the precomputed values $m_i$ associated to the breakpoints.
Consequently, we can construct a data structure of size $O(m)$ to answer such submatrix
maximum queries in $O(\pred{m,n}+\pred{n,m})$ time.

The above solution can be generalized to queries that start at the first row of $M$ but do not necessarily end at the last row of $M$. This is done
by considering the Monge matrices $M_{i}$ consisting
of the first $i$ rows of $M$. For every such matrix, we need a predecessor structure
storing all of its breakpoints, and additionally a range maximum structure over
their associated values $m_i$. Hence now we need to construct a similar
tree $T$ as in Lemma~\ref{lem:subcolumn} on $O(m)$ nodes,
but now every node has both a weight and a value. The weight of a node is the column number
of the corresponding breakpoint $c_{k}$, and the value is its $m_{k}$ (or undefined if $k=1$).
As in Lemma~\ref{lem:subcolumn}, the breakpoints of $M_{i}$ are exactly the ancestors of the node $s(i)$. Note that
every $m_{k}$ is defined in terms of $c_{k-1}$ and $c_{k}$, but this is not a problem because
the predecessor of a breakpoint does not change during the whole construction.
We maintain a weighted ancestor structure
using the weights (in order to find $c_{i}$ and $c_{i'}$ in $O(\pred{m,n})$ time), and a {\em generalized range maximum structure} using the values. A generalized range maximum structure of a tree $T$, given two query nodes $u$ and $v$, returns the maximum value on the unique $u$-to-$v$
path in $T$. It can be implemented
in $O(m)$ space and $O(1)$ query time after $O(m\log m)$ preprocessing~\cite{DemaineLandauWeimann}
once we have the values. The values can be computed with subrow maximum queries
in $O(m \cdot \pred{n,m}) = O(m\log m)$ total time.
\qed \end{proof}

By applying Lemma~\ref{lem:submatrix} twice to every node of the full binary tree $\mathcal T$,
we construct in $O(m\log ^2 m)$ time a data structure of size $O(m\log m)$ to
answer submatrix maximum queries in $O(\pred{m,n}+\pred{n,m})$ time. In
order to apply Lemma~\ref{lem:submatrix} to a node of $\mathcal T$ we
need a subrow maximum query data structure for the corresponding rows
of the matrix $M$. Note, however, that a single subrow maximum query
data structure for $M$ can be used for all nodes of $\mathcal T$. We
thus obtained the following theorem.

\begin{theorem}
\label{thm:submatrix}
Given an $m\times n$ Monge matrix $M$, and a data structure 
answering subrow maximum queries on $M$ in $O(\pred{n,m})$ time, one can
construct in $O(m\log ^2 m)$ time  a data structure
taking $O(m \log m)$
additional space, that answers submatrix maximum queries on $M$ in $O(\pred{m,n}+\pred{n,m})$ time.
\end{theorem}

By combining Theorem~\ref{thm:subcolumn} with Theorem~\ref{thm:submatrix},
given an $n\times n$ Monge matrix $M$, a data structure of size $O(n\log n)$ can be
constructed in $O(n\log^2 n)$ time to answer submatrix maximum queries in
$O(\pred{n,n})$ time.

\section{Obtaining linear space}
\label{sec:linear}
In this section we show how to decrease the space of the data structure presented in Section~\ref{sec:structure} to be linear.
We extend the idea developed in our previous paper~\cite{ourICALP}.
The previous linear space solution was based on partitioning the matrix  $M$ into $n/x$ matrices $M_{1},M_{2},\ldots,M_{n/x}$, where each $M_{i}$ is a {\em slice} of $M$ consisting of $x=\log n$ consecutive rows. 
Then, instead of working with the matrix $M$, we worked with the $(n/x)\times n$ matrix $M'$, where $M'[i,j]$ is the maximum entry in the $j$-th column of $M_{i}$.

\paragraph{\bf Subcolumn queries.}
Consider a subcolumn query. Suppose the query is entirely contained in some $M_{i}$. This means it spans less than $x=\log n$ rows. In~\cite{ourICALP}, since the desired query time was $O(\log n)$, a query  simply inspected all elements of the subcolumn. In our case however, since the desired query time is only $O(\log \log n)$, 
we apply the above partitioning scheme twice. We explain this now.

We start with the following lemma, that provides an efficient data structure for  queries consisting of a single column and {\em all} rows in rectangular matrices.


\begin{lemma}[the micro data structure] \label{lemma:micro}
Given an $x\times n$ TM matrix and $r >0$, one can construct in $O(x \log n /\log r)$ time, a data structure of size $O(x)$ that given a query column can report the maximum entry in the entire column in 
$O(r+pred(x,n))$ time. 
\end{lemma}
\begin{proof}
Out of all $n$ columns of the input matrix $M$, we will designate
$O(x)$ columns as {\em special} columns. For each of these special columns we will eventually compute its maximum element.  
The first $x$ special columns of $M$ are  columns $1, n/x, 2n/x, 3n/x,\ldots, n$ and are denoted $j_1,\ldots, j_x$. 

Let $X$ denote the $x \times x$ submatrix obtained by taking all $x$
rows but only the $x$ special columns $j_1,\ldots, j_x$. It is easy to verify that $X$
is TM. We can therefore run the SMAWK
algorithm~\cite{SMAWK} on $X$ in  $O(x)$ time and obtain the column
maxima of all special columns. 
Let $r(j)$ denote the row containing the maximum element in column $j$~\footnote{We assume that no elements of the matrix are equal. The ties are resolved lexicographically.}
Since $M$ is TM, the $r(j)$ values are monotonically
non-decreasing. Consequently, $r(j)$ of a non-special column $j$ must be
between $r(j_i)$ and $r(j_{i+1})$ where $j_i<j$ and $j_{i+1}>j$ are
the two special  columns bracketing $j$  
 (see Figure~\ref{fig}). 
 
\begin{figure}[h!]
   \centering
   \includegraphics[scale=0.4]{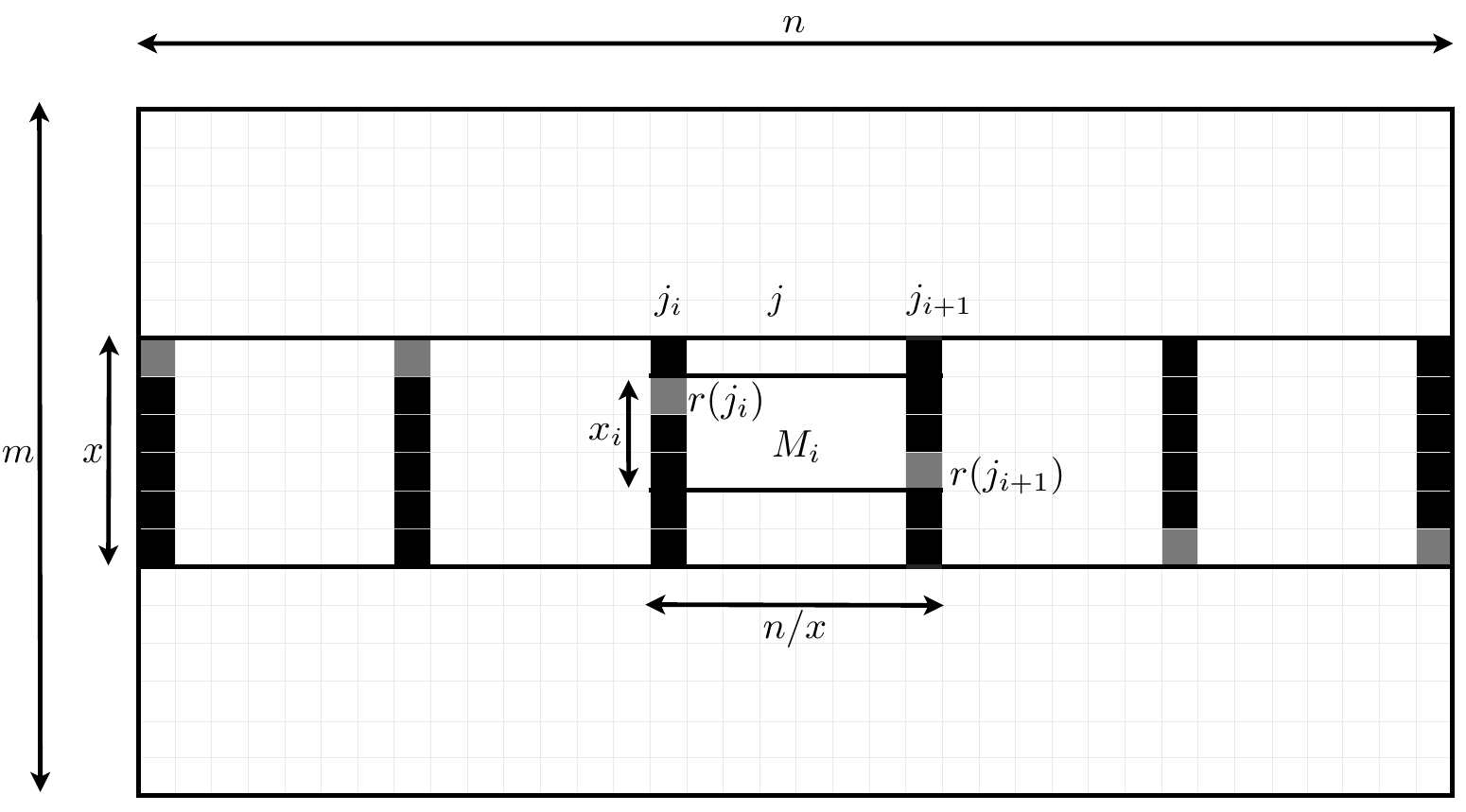}
   \caption{An $x\times n$ matrix inside an $m\times n$ matrix. The
     black columns are the first $x$ special columns. The
     (monotonically  non-decreasing) gray cells inside these special
     columns are the column maxima (i.e., the $r(j_i)$ values of breakpoints $j_i$). The maximum element of column $j$ in the $x\times n$ matrix must be between $r(j_i)$ and $r(j_{i+1})$ (i.e., in matrix $M_i$).}
  \label{fig}
 \end{figure}

For every $i$, let $x_i = r(j_{i+1}) - r(j_i)$. If $x_i  \le r$ then
{\em no} column between $j_i$ and $j_{i+1}$ will ever be a special
column. When we will query such a column $j$ we can simply check
(at query-time) the $r$ elements of $j$ between rows $r(j_i)$ and
$r(j_{i+1})$ in $O(r)$ time. If, however, $x_i  > r$, then we designate
more special columns between $j_i$ and $j_{i+1}$. This is done
recursively on the $x_i  \times (n/x)$ matrix $M_i$ composed of
rows $r(j_i),\ldots, r(j_{i+1})$ and columns $j_i,\ldots,
j_{i+1}$. That is, we mark $x_i$ evenly-spread columns of $M_i$ as
special columns,  and run SMAWK in $O(x_i)$ time on the $x_i \times
x_i$ submatrix $X_i$   obtained by taking all $x_i$ rows but only
these $x_i$ special columns. We continue recursively until either $x_i
\le r$ or the number of columns in $M_i$ is at most $r$. In the
latter case, before terminating, the recursive call runs SMAWK in $O(x_i+r)=O(x_i)$ time on the $x_i
\times r$ submatrix $X_i$ obtained by taking the $x_i$ rows and {\em
  all} columns of $M_i$ (i.e., all columns of $M_i$ will become
special). 

After the recursion terminates, every column $j$ of $M$ is either
special (in which case we computed its maximum), or its maximum  is
known to be in one of at most $r$ rows (these rows are specified by the $r(\cdot)$
values of the two special columns bracketing $j$). 
Let $s$ denote the total number of columns that are marked as special. 
We claim that $s = O(x \log n /\log r)$. To see this, notice that the
number of columns in every recursive call decreases by a factor of at
least $r$ and so the recursion depth is $O(\log_r n) = O(\log n /\log
r)$. In every recursive level, the number of added special columns is
$\sum x_i$ over all $x_i's$ in this level that are at least $r$. In
every recursive level, this sum is bounded by $2x$ because each one of
the $x$ rows of $M$ can appear in at most two $M_i$'s  (as the last
row of one and the first row of the other). Overall, we get $2x \cdot
O(\log n /\log r) = O(x \log n /\log r)$. 

Notice that $s = O(x \log n /\log r)$ implies that the total time
complexity of the above procedure is also $O(x \log n /\log r)$. This
is because whenever we run SMAWK on a $y \times y$ matrix it takes
$O(y)$ time and $y$ new columns are marked as special.  To complete
the construction, we go over the $s$ special columns from left to right in
$O(s)$ time and throw away (mark as non-special) any column whose
$r(\cdot)$ value is the same as that of the preceding special column. This
way we are left with only $O(x)$ special columns, and the difference in
$r(\cdot)$ between consecutive special columns is at least $1$ and at most
$r$. In fact, it is easy to maintain $O(x)$ (and not $O(s)$) space
{\em during} the construction by only recursing on sub matrices $M_i$
where $x_i >1$. We note that when $r=1$, the eventual special columns are
exactly the set of breakpoints of the input matrix $M$.

The final data structure is a predecessor data structure that holds
the $O(x)$ special columns and their associated $r(\cdot)$
values. Upon query of some column $j$, we search in $pred(x,n)$ time
for the predecessor and successor of $j$ and obtain the two $r(\cdot)$
values. We then  search for the maximum of column $j$ by explicitly
checking all the (at most $r$) relevant rows of column $j$. The query time is
therefore  $O(r+pred(x,n))$ and the space $O(x)$.  
\qed \end{proof}

In the case of $x = O(\log n)$, using atomic heaps~\cite{FredmanW94} (which support
predecessor searches in constant time) we obtain the following corollary:

\begin{corollary}\label{lem:micro}
Given an $x\times n$ TM matrix, a data structure of size $O(x)$ can be constructed in
$O(x \log n)$ time to answer entire-column maximum queries  in $O(1)$ time, if $x=O(\log n)$.
\end{corollary}

It is possible to use Lemma~\ref{lemma:micro} to obtain a subcolumn data
structure with faster $O(n \log n / \log\log n)$ preprocessing time,
at the cost of slower $O(\log n)$ query time (cf.~\cite[Lemma
2]{ourICALP}).  We next describe our new subcolumn data structure, which uses the above corollary and two applications of
the partitioning scheme.

\begin{theorem}
\label{thm:subcolumn2}
Given an $m\times n$ Monge matrix $M$, a data structure of size $O(m)$ can be constructed
in $O(m\log n)$ time to answer subcolumn maximum queries in $O(\log\log (n+m))$ time.
\end{theorem}

\begin{proof}
We first partition $M$ into $m/x$ matrices
$M_{1},M_{2},\ldots,M_{m/x}$, where $x=\log m$.
Every $M_{i}$ is a slice of $M$ consisting of $x$ consecutive rows. Next,  we partition
every $M_{i}$ into $x/x'$ matrices $M_{i,1},M_{i,2},\ldots,M_{i,x'}$, where $x'=\log\log m$.
Every $M_{i,j}$ is a slice of $M_{i}$ consisting of $x'$ consecutive rows (without loss of generality, assume that $x$ divides $m$ and $x'$ divides $x$).
Now we define a new $(m/x)\times n$ matrix $M'$, where $M'[i,j]$ is the maximum
entry in the $j$-th column of $M_{i}$. Similarly, for every $M_{i}$ we define
a new $(x/x')\times n$ matrix $M'_{i}$, where $M'_{i}[j,k]$ is the maximum entry
in the $k$-th column of $M_{i,j}$.

We apply Corollary~\ref{lem:micro} on every $M_{i}$ and $M_{i,j}$ in $O(m\log n)$ total time
and $O(m)$ total space, so that any $M'[i,j]$ or $M'_{i}[j,k]$ can be retrieved in $O(1)$ time. Furthermore,
it can be easily verified that $M'$ and all $M'_{i}$s are also Monge. 
To prove this, it is enough to argue that if $N$ is an $4\times 2$
Monge matrix,
the $2\times 2$ matrix $N'$ created by partitioning $N$ into two
slices, each consisting
of two rows, whose elements are the maxima in every column of each slice, is also
Monge. To this end, we need to compare:
$$N'[1,1]+N'[2,2]=\max(N[1,1],N[2,1])+\max(N[3,2],N[4,2])$$
and
$$N'[1,2]+N'[2,1]=\max(N[1,2],N[2,2])+ \max(N[3,1],N[4,1]).$$
Let $\max(N[1,2],N[2,2])=N[i,2]$, where $i\in \{1,2\}$, and similarly
$\max(N[3,1],N[4,1])=N[i',1]$, where $i'\in \{3,4\}$. Then
$$(N'[1,1]+N'[2,2]) - (N'[1,2]+N'[2,1])  \geq (N[i,1] + N[i',2]) - (N[i,2]+N[i',1])$$
which is at least $0$ because of $N$ being Monge.
  
Therefore, because $M'$ and all $M'_{i}$ are all Monge, and
by Corollary~\ref{lem:micro} their entries can be accessed in $O(1)$ time, 
we can apply Theorem~\ref{thm:subcolumn} on $M'$
and every $M'_{i}$. The total construction time is 
$O((m/x)\log(m/x)\log n+(m/x)(x/x')\log(x/x')\log n)=O(m\log n)$, 
and the total size of all structures constructed so far is 
$O((m/x)\log(m/x)+(m/x)(x/x')\log(x/x'))=O(m)$.

Now consider a subcolumn maximum query. If the range of rows is fully within a single
$M_{i,j}$, the query can be answered naively in $O(x')=O(\log\log m)$ time.
Otherwise, if the range of rows is fully within a single $M_{i}$, the query can be decomposed
into a prefix fully within some $M_{i,j}$, an infix corresponding to a range of rows
in $M'_{i}$, and a suffix fully within some $M_{i,j'}$. The maximum in the prefix and the suffix can
be computed naively in $O(x')=O(\log\log m)$ time, and the maximum in the infix
can be computed in $O(\log\log n)$ time using the structure constructed for $M'_{i}$.
Finally, if the range of rows starts inside some $M_{i}$ and ends inside another $M_{i'}$,
the query can be decomposed into two queries fully within $M_{i}$ and $M_{i'}$, respectively,
which can be processed in $O(\log\log n)$ time as explained before, and an infix
corresponding to a range of rows of $M'$. The maximum in the infix can be computed
in $O(\log\log n)$ time using the structure constructed for $M'$.
\qed \end{proof}

\paragraph{\bf Submatrix queries.}
We are ready to present the final version of our data structure. It is based on two
applications of the partitioning scheme, and an additional trick of transposing the matrix.

\begin{theorem}
\label{thm:submatrix2}
Given an $n\times n$ Monge matrix $M$, a data structure of size $O(n)$ can be constructed
in $O(n\log n)$ time to answer submatrix maximum queries in $O(\log\log n)$ time.
\end{theorem}

\begin{proof}
We partition $M$ as described in the proof of Theorem~\ref{thm:subcolumn2}, i.e., $M$ is
partitioned into $n/x$ matrices $M_{1},M_{2},\ldots,M_{n/x}$, where $x=\log n$, and every $M_{i}$
is then partitioned into $x/x'$ matrices $M_{i,1},M_{i,2},\ldots,M_{i,x'}$, where $x'=\log\log n$.
Then we define smaller Monge matrices $M'$ and $M'_{i}$, and provide $O(1)$ time access to their
entries with Corollary~\ref{lem:micro}. We apply
Theorem~\ref{thm:subcolumn2} to the transpose of $M'$ to get a subrow
maximum query data structure for $M'$. This takes $O(n)$ space and
$O(n\log n)$ time. With this data structure we can apply Theorem~\ref{thm:submatrix} on
$M'$, which takes an additional $O(\frac{n}{\log n} \log \frac{n}{\log
  n}) = O(n)$ space and $O(n \log n)$ time. 
We also apply Theorem~\ref{thm:subcolumn2} to the transpose of the
$\frac{n}{\log\log n}$-by-$n$
matrix obtained by stacking the $\frac{n}{\log n}$ $M'_i$ matrices. 
This takes $O(n)$ space and
$O(n\log n)$ time. This serves as a subrow maximum data structure for
each $M'_i$, so we can apply Theorem~\ref{thm:submatrix} to each
$M'_i$ separately, which takes a total of  $O(\frac{n}{\log n}
\frac{\log n}{\log \log n}\log(\frac{\log n}{\log\log n})) = O(n)$ additional space and $O(n \log \log n)$ time. 


We repeat the above preprocessing on the transpose of $M$.
Now consider a submatrix maximum query. If the range of rows starts inside some $M_{i}$ and 
ends inside another $M_{i'}$, the query can be decomposed into two queries fully within $M_{i}$ and
$M_{i'}$, respectively, and an infix corresponding to a range of rows of $M'$. The maximum
in the infix can be computed in $O(\log\log n)$ time using the structure constructed
for $M'$. Consequently, it is enough to show how to answer a query in $O(\log\log n)$
time when the range of rows is fully within a single $M_{i}$. In such case, if the range of rows
starts inside some $M_{i,j}$ and ends inside another $M_{i,j'}$, the query can be decomposed
into a prefix fully within $M_{i,j}$, an infix corresponding to a range of rows in $M'_{i}$
and a suffix fully within some $M_{i,j'}$. 
The query on the infix can be answered using the data structure for $M'_{i}$. 
Consequently, we reduced the  query  in $O(\log\log n)$ time to four
queries such that the range of rows in each query is fully
within a single $M_{i,j}$.
Since each $M_{i,j}$ consists of $O(\log\log n)$ rows of $M$,
by taking the union of the rows of $M$ corresponding to all
these $M_{i,j}$'s and also including the row containing the maximum in the infixes,
we have identified, in $O(\log\log n)$ time, a set of $O(\log\log n)$ rows of
$M$ that contain the desired submatrix maximum. 

Now we repeat the same procedure on the transpose of $M$ to identify a set of
$O(\log\log n)$ columns of $M$ that contain the desired submatrix
maximum.
Since a submatrix of a Monge matrix is also Monge, the submatrix of
$M$ corresponding to these sets of candidate rows and columns is an
$O(\log\log n) \times O(\log\log n)$ Monge matrix. 
By running the SMAWK algorithm~\cite{SMAWK} in $O(\log\log n)$ time on
this small
Monge matrix, we can finally determine the answer.
\qed \end{proof}

\section{Lower Bound}
\label{sec:lower bound}
A predecessor structure stores a set of $n$ integers $S\subseteq [0,U)$, so that given $x$ we can determine the largest $y\in S$
such that $y\leq x$. As shown by P{\v{a}}tra{\c{s}}cu and Thorup~\cite{PT2006}, for $U=n^{2}$
any predecessor structure consisting of $O(n \polylog(n))$ words needs $\Omega(\log\log n)$ time to answer queries,
assuming that the word size is $\Theta(\log n)$. We will use their result to prove that our structure is in fact optimal.

Given a set of $n$ integers $S\subseteq [0,n^{2})$ we want to construct an $n\times n$ Monge matrix $M$ such that the predecessor of
any $x$ in $S$ can be found using one submatrix maximum query on $M$ and $O(1)$ additional time (to decide which query to ask
and then return the final answer). Then, assuming that for any $n\times n$ Monge matrix there exists a data structure of size $O(n\polylog(n))$
answering submatrix maximum queries in $o(\log\log n)$ time, we can construct a predecessor structure
of size $O(n\polylog(n))$ answering queries in $o(\log\log n)$ time, which is not possible.
The technical difficulty here is twofold. First, $M$ should be Monge. Second, we are working in the indexing model, i.e., the
data structure for submatrix maximum queries should be able to access the matrix. Therefore, for the lower bound to carry over, $M$ should have the following
property: there is a data structure of size $O(n\polylog(n))$ which retrieves any $M[i,j]$ in $O(1)$ time. Guaranteeing
that both properties hold simultaneously is not trivial. 

Before we proceed, let us comment on the condition $S\subseteq [0,n^{2})$. While quadratic universe is enough to invoke
the $\Omega(\log\log n)$ lower bound for structures of size
$O(n\polylog(n))$, our reduction actually implies  that even for larger
polynomially bounded universes, i.e., $S\subseteq [0,n^{c})$, for any
fixed $c$, it is possible to construct an $n\times n$ Monge matrix $M$
such that the predecessor of $x$ in $S$ can be found with $O(1)$ submatrix maximum queries on $M$ and $O(1)$
additional time (and, as previously, any $M[i,j]$ can be retrieved in
$O(1)$ time with a structure of size $O(n)$). This is a consequence of the following lemma.

\begin{lemma}
\label{lem:universe reduction}
For any constant $c\geq 2$, predecessor queries on a set of $n$ integers $S\subseteq [0,n^{c})$ can be reduced in $O(1)$ time to $O(1)$
predecessor queries on a set of $n$ integers $S'\subseteq [0,n^{2})$ with a structure of size $O(n)$.
\end{lemma}

\begin{proof}
First we describe a weaker version of the reduction for $c=4$, where the resulting set of integers
is $S'\subseteq [0,3n^2)$.

Let $S=\{x_{1},x_{2},\ldots,x_{n}\}$. We represent every $x_{i}$ in base $n^{2}$ as $x_{i}=y_{i}\cdot n^{2}+z_{i}$, where
$y_{i},z_{i}\in [0,n^{2})$. We create a new set $Y \subseteq [0,n^{2})$ storing all $y_{i}$s and
a new set $Z\subseteq [0,n^2)$ storing all $z_{i}$s. For any $t\in [0,n^{2})$, let
$\text{rank}_{Y}(t)$ and $\text{rank}_{Z}(t)$ denote the rank of $t$ in $Y$ and $Z$, respectively,
where rank is the number of smaller elements in the set.
We create another set $R\subseteq [0,n^{2})$ storing 
elements of the form $\text{rank}_Y(y_{i})\cdot n+\text{rank}_Z(z_{i})$.
We also create a perfect hash table of size $O(n)$ mapping $y_{i}$ to $\text{rank}_{Y}(y_{i})$
and $z_{i}$ to $\text{rank}_{Z}(z_{i})$.
To find the predecessor of $x$ in $S$, we first represent it as $x=y\cdot n^2+z$.
We claim that it is always possible to reduce locating the predecessor
of $x$ in $S$ to the case where $y\in Y$ and $z\in Z$ in two steps.
Let $y'$ denote the predecessor of $y$ in $Y$ and $z'$ denote the predecessor of $z$ in $Z$.

\begin{enumerate}
\item If $z'$ is not defined, we decrease $y$ by one (adjusting $y'$ if necessary)
and replace $z$ by the largest element of $Z$. Otherwise, we replace $z$ by $z'$.
\item If $y'$ is not defined, $x$ has no predecessor in $S$. Otherwise, if $y' \neq y$ we replace
$y$ by $y'$ and $z$ by the largest element of $Z$.
\end{enumerate}

Both steps maintain the predecessor of $x$ in $S$ and take $O(1)$ time.
Finally, having reduced the general case so that $y\in Y$ and $z\in Z$, we
locate the predecessor of $x'=\text{rank}_{Y}(y)\cdot n+\text{rank}_{Z}(z)$ in $R$.
Because $y\in Y$ and $z\in Z$, both $\text{rank}_{Y}(y)$ and $\text{rank}_{Z}(z)$ can be
retrieved in $O(1)$ time from the perfect hash tables.
The predecessor of $x'$ in $R$ corresponds to the predecessor of $x$ in $S$,
because comparing two elements of the same set is equivalent to comparing their ranks there.
Formally, $x_i \leq x$ iff $y_i < y$ or $y_i=y$ and $z_i \leq z$, which is equivalent
to $\text{rank}_{Y}(y_i) < \text{rank}_{Y}(y)$ or $\text{rank}_{Y}(y_i)=\text{rank}_{Y}(y)$
and $\text{rank}_{Z}(z_i) \leq \text{rank}_{Z}(z)$, which because the ranks are all from $[0,n)$
can be stated as $\text{rank}_{Y}(y_i)\cdot n+\text{rank}_{Z}(z_i) \leq \text{rank}_{Y}(y)\cdot n
+\text{rank}_{Z}(z)$.
Consequently, a predecessor query on $S$ can be reduced into one predecessor query
into each of $Y,Z,R$. These three sets can be combined into a single set $S' \subseteq [0,3n^2)$,
such that predecessor queries in either of them
can be answered with predecessor queries on $S'$, by simply shifting every element of $Z$ by $n^{2}$
and every element of $R$ by $2n^{2}$.
Finally, the size of $S'$, which is up to $3n$ right now, can be reduced to $n$ as follows.
Let the elements of $S'$ be $x_{1}<x_{2}< \ldots < x_{3n}$. We store every $x_{3i}$,
for $i=1,2,\ldots,n$ in the predecessor structure. Additionally, for every $i$ we
explicitly store $x_{3i+1}$ and $x_{3i+2}$. Knowing the predecessor $x_{i}$ of $x$ among the
chosen elements allows us to find its predecessor among all elements in $O(1)$ time by additionally
inspecting $x_{3i+1}$ and $x_{3i+2}$.

Now we explain how to extend the above reduction for any constant $c\geq 2$, while also
ensuring that the resulting set of integers is $S'\subseteq [0,n^2)$.
If $n < 5$, we answer predecessor queries naively in $O(1)$ time.
If $n\geq 5$, by modifying the above reduction so that every $x_i$ is represented as $y_i \cdot n^2+z_i$,
where $y_i\in[0,n^{c-2})$ and $z_i\in [0,n^2)$, we obtain a set of $n$ integers from $[0,n^{c-1})$.
Hence, by iterating $c-3$ times we finally obtain a set of $n$ integers from $[0,n^3)$.
Then one final iteration, where we represent every $x_i$ as $y_i\cdot B + z_i$,
with $y_i,z_i\in[0,B)$ with $B=\lceil n^{1.5} \rceil$,
allows us to reduce the size of the universe to $2B+n^{2}$, which is at most $2n^{2}$
for $n\geq 5$.
To reduce the size of the universe to $n^{2}$, we divide every $x_{i}$ by 2.
Let the resulting set be $x'_{1}<x'_{2}<\ldots x'_{n'}$. 
We store a perfect hash table mapping $x'_{i}$ to $x'_{i-1}$ (if it exists) and a list of elements $x_{j}$ such that
$\lfloor x_{j}/2\rfloor = x'_{i}$ (note that there are at most two such $j$s).
To find the predecessor of $x$, we then find the predecessor $x'_{i}$ of $\lfloor x/2\rfloor$ in the obtained set.
Then, we inspect all the elements stored on the lists of $x'_{i}$ and $x'_{i-1}$ (accessed from the entry of $x'_{i}$ in the
perfect hash table) and return the largest not exceeding $x$ in $O(1)$ time.
\qed \end{proof}

\noindent The following propositions are easy to verify:

\begin{proposition}
\label{prop:adjacent condition}
An $m\times n$ matrix $M$ is Monge iff $M[i,j]+M[i+1,j+1] \ge M[i+1,j]+M[i,j+1]$ for all $i=1,2,\ldots,m-1$ and $j=1,2,\ldots,n-1$.
\end{proposition}

\begin{proposition}
\label{prop:adjustment}
If a matrix $M$ is Monge, then for any vector $H$ the matrix $M'$, where $M'[i,j]=M[i,j]+H[j]$ for all $i,j$, is also Monge.
\end{proposition}

\begin{proposition}
\label{prop:undef}
If a matrix $M$ is partial Monge, then it remains partial Monge after replacing any element of $M$ by a blank, so long as the defined entries in each row and in each column remain contiguous.
\end{proposition}

\begin{proposition}
\label{prop:duplicate}
If a $m$-by-$n$ matrix $M$ is (partial) Monge, then the $(m+1)$-by-$n$ matrix resulting by replacing any row of $M$ by two identical copies of that row is also (partial) Monge. An analogous statement holds for duplicating any column of $M$. 
\end{proposition}

\begin{theorem}
\label{thm:reduction}
For any set of $n$ integers $S\subseteq [0,n^{2})$, there exists a data structure of size $O(n)$ returning any $M[i,j]$
in $O(1)$ time, where $M$ is a Monge matrix such that the predecessor of $x$ can be found using $O(1)$ time and one
submatrix maximum query on $M$.

\end{theorem}

\begin{proof}

We partition the universe $[0,n^{2})$ into $n$ parts $[0,n),[n,2n),\ldots$. The $i$-th part $[i\cdot n,(i+1)\cdot n)$
defines a Monge matrix $M_{i}$ consisting of $2+|S\cap [i\cdot
n,(i+1)\cdot n)|$ rows and $n$ columns. The first and the last row are artificial, and others encode the
elements of $S\cap [i\cdot n,(i+1)\cdot n)$.
The idea is to encode the predecessor of $x\in [0,n^{2})$ by the
maximum element in the $(x\bmod n+1)$-th column
of $M_{\lfloor x/n\rfloor }$. 
 We first describe
how these matrices are defined, and then show how to stack them together.

Consider any $0 \leq i < n$. Every element in $S\cap
[i\cdot n,(i+1)\cdot n)=\{a_{1},a_{2},\ldots,a_{k}\}$ has a unique
corresponding row in $M_i$.
Let $a_{j}=i\cdot n + a'_{j}$, so that $a'_{1}<a'_{2}<\ldots < a'_{k}$
and $a'_{j}\in [0,n)$ for all $j$, and also define $a'_{k+1}=n$. We
describe an incremental construction of $M_i$.
For technical reasons, we start with an artificial top row containing $n-1,n-2,\ldots,1$. Then we add the rows corresponding
to $a'_{1},a'_{2},\ldots,a'_{k}$. The row corresponding to $a'_{j}$ consists of three parts. The middle part starts at the
$(a'_{j}+1)$-th column, ends at the $a'_{j+1}$-th column, and contains only $n$'s. The elements in the left part 
increase by $1$ and end with $n-1$ at the $a'_j$-th column, similarly
the elements in the right part (if any) start with $n-1$ at the
$(a'_{j+1} + 1)$-th column and decrease by $1$.
Formally, the $k$-th element of the $(j+1)$-th row, denoted $M_{i}[j+1,k]$, is defined as follows.
\begin{eqnarray}
\label{eqn:M_i}
M_{i}[j+1,k]=\begin{cases}
n-1-a'_{j}+k &\mbox{ if } k\in [1,a'_{j}] \\
n &\mbox{ if } k\in [a'_{j}+1,a'_{j+1}] \\
n-k+a'_{j+1} &\mbox{ if } k\in [a'_{j+1}+1,n]
\end{cases}
\end{eqnarray}

Finally, we end with an artificial bottom row containing $1,2,\ldots,n$.
See Figure~\ref{fig:reduction} for an example. We need to argue that every $M_{i}$ is Monge. By
Proposition~\ref{prop:adjacent condition}, it is enough to consider every pair of adjacent rows $r_{1},r_{2}$ there.
Define $r'_{1}[j]=r_{1}[j]-r_{1}[j-1]$ and similarly $r'_{2}[j]=r_{2}[j]-r_{2}[j-1]$. To prove that $M_{i}$ is Monge, it is
enough to argue that $r'_{2}[j]\geq r'_{1}[j]$ for all $j\geq 2$. By construction, both $r'_{1}$ and $r'_{2}$ are of the form
$1,1,\ldots,1,0,0,\ldots,0,-1,-1,\ldots,-1$, and all $0$'s in $r'_{2}$ are on the right of all $0$'s in $r'_{1}$.
Therefore, $M_{i}$ is Monge.

Now one can observe that the predecessor of $x\in [0,n^{2})$ can be found by looking at the $(x\bmod n+1)$-th column
of $M_{\lfloor x/n\rfloor }$. We check if $x < a_{1}$, and if so return the predecessor of $a_{1}$ in the whole $S$.
This can be done in $O(1)$ time and $O(n)$ additional space by explicitly storing $a_{1}$ and its predecessor for every $i$.
Otherwise we know that the predecessor of $x$ is $a_{j}$ such that $x\bmod n \in [a'_{j},a'_{j+1})$, and, 
by construction, we only need to find $j\in [1,k]$ such that the $(x\bmod n+1)$-th element of row $j+1$ in $M_{i}$ is $n$.
This is exactly a subcolumn maximum query.

We cannot simply concatenate all $M_{i}$'s to form a larger Monge matrix. We use Proposition~\ref{prop:adjustment}
instead. Initially, we set $M=M_{0}$. Then we consider every other $M_{i}$ one-by-one maintaining invariant
that the current $M$ is Monge and its last row is $1,2,\ldots,n$. In every step we add the vector 
$H=[n-1,n-3,\ldots,-n+1]$ to the current matrix $M$, obtaining a matrix $M'$ whose last row is $n,n-1,\ldots,1$. By Proposition~\ref{prop:adjustment}, $M'$ is Monge. 
Then we can construct the new $M$ by appending $M_{i}$ without its first row to $M'$.
Because the first row of $M_{i}$ is also $n-1,n-2,\ldots,1$, the new $M$ is also Monge. Furthermore, because we add the
same value to all elements in the same column of $M_{i}$, answering subcolumn maximum queries on $M_{i}$ can
be done with subcolumn maximum queries on the final $M$.  The right side of Figure~\ref{fig:reduction} depicts
the final Monge matrix $M$.

We need to argue that elements of $M$ can be accessed in $O(1)$ time using a data structure of size $O(n)$. To
retrieve $M[j,k]$, first we lookup in $O(1)$ time the appropriate $M_{i}$ from which it originates. This
can be preprocessed and stored for every $j$ in $O(n)$ total space and allows us to reduce the question to retrieving
$M_{i}[j',k]$. Because Proposition~\ref{prop:adjustment} is applied exactly $n-1-i$ times after appending $M_{i}$
to the current $M$, then we can return $M_{i}[j',k]+(n-1-i)\cdot H[k]$. To find $M_{i}[j',k]$, we just directly use Equation~\ref{eqn:M_i},
which requires only storing $a'_{1},a'_{2},\ldots,a'_{n}$ in $O(n)$ total space.
\qed \end{proof}

\begin{figure}[htb]
\centering
{\small
\begin{minipage}[]{0.2\linewidth}
$$ M_{0} = \left[ \begin{array}{cccccccc}
8 & 7 & 6 & 5 & 4 & 3 & 2 & 1\\
5 & 6 & 7 & 8 & 8 & 8 & 8 & 8\\
1 & 2 & 3 & 4 & 5 & 6 & 7 & 8\\
\end{array} \right] $$
$$ M_{1} = \left[ \begin{array}{cccccccc}
8 & 7 & 6 & 5 & 4 & 3 & 2 & 1\\
1 & 2 & 3 & 4 & 5 & 6 & 7 & 8\\
\end{array} \right] $$
$$ M_{2} = \left[ \begin{array}{cccccccc}
8 & 7 & 6 & 5 & 4 & 3 & 2 & 1\\
6 & 7 & 8 & 8 & 8 & 7 & 6 & 5\\
3 & 4 & 5 & 6 & 7 & 8 & 7 & 6\\
2 & 3 & 4 & 5 & 6 & 7 & 8 & 8\\
1 & 2 & 3 & 4 & 5 & 6 & 7 & 8\\
\end{array} \right] $$
$$ M_{3} = \left[ \begin{array}{cccccccc}
8 & 7 & 6 & 5 & 4 & 3 & 2 & 1\\
1 & 2 & 3 & 4 & 5 & 6 & 7 & 8\\
\end{array} \right] $$
\end{minipage}
\hspace{0.5cm}
\begin{minipage}[]{0.2\linewidth}
$$ M_{4} = \left[ \begin{array}{cccccccc}
8 & 7 & 6 & 5 & 4 & 3 & 2 & 1\\
1 & 2 & 3 & 4 & 5 & 6 & 7 & 8\\
\end{array} \right] $$
$$ M_{5} = \left[ \begin{array}{cccccccc}
8 & 7 & 6 & 5 & 4 & 3 & 2 & 1\\
6 & 7 & 8 & 8 & 8 & 8 & 7 & 6\\
2 & 3 & 4 & 5 & 6 & 7 & 8 & 8\\
1 & 2 & 3 & 4 & 5 & 6 & 7 & 8\\
\end{array} \right] $$
$$ M_{6} = \left[ \begin{array}{cccccccc}
8 & 7 & 6 & 5 & 4 & 3 & 2 & 1\\
1 & 2 & 3 & 4 & 5 & 6 & 7 & 8\\
\end{array} \right] $$
$$ M_{7} = \left[ \begin{array}{cccccccc}
8 & 7 & 6 & 5 & 4 & 3 & 2 & 1\\
7 & 8 & 8 & 8 & 7 & 6 & 5 & 4\\
4 & 5 & 6 & 7 & 8 & 8 & 8 & 8\\
1 & 2 & 3 & 4 & 5 & 6 & 7 & 8\\
\end{array} \right] $$
\end{minipage}
\hspace{2cm}
\begin{minipage}[]{0.3\linewidth}
$$ M = \left[ \begin{array}{
cccccccc}
57 & 42 & 27 & 12 & -3 & -18 & -33 & -48\\
54 & 41 & 28 & 15 & 1 & -13 & -27 & -41\\
50 & 37 & 24 & 11 & -2 & -15 & -28 & -41\\
43 & 32 & 21 & 10 & -1 & -12 & -23 & -34\\
41 & 32 & 23 & 13 & 3 & -8 & -19 & -30\\
38 & 29 & 20 & 11 & 2 & -7 & -18 & -29\\
37 & 28 & 19 & 10 & 1 & -8 & -17 & -27\\
36 & 27 & 18 & 9 & 0 & -9 & -18 & -27\\
29 & 22 & 15 & 8 & 1 & -6 & -13 & -20\\
22 & 17 & 12 & 7 & 2 & -3 & -8 & -13\\
20 & 17 & 14 & 10 & 6 & 2 & -3 & -8\\
16 & 13 & 10 & 7 & 4 & 1 & -2 & -6\\
15 & 12 & 9 & 6 & 3 & 0 & -3 & -6\\
8 & 7 & 6 & 5 & 4 & 3 & 2 & 1\\
7 & 8 & 8 & 8 & 7 & 6 & 5 & 4\\
4 & 5 & 6 & 7 & 8 & 8 & 8 & 8\\
1 & 2 & 3 & 4 & 5 & 6 & 7 & 8\\
\end{array} \right] $$
\end{minipage}
\caption{Reduction for $n=8$ and $S=\{8\cdot 0+3,8\cdot 2+2,8\cdot 2+5,8\cdot 2+6,8\cdot 5+2,8\cdot 5+6,8\cdot 7+1,8\cdot 7+4\}$.}
\label{fig:reduction}
}
\end{figure}

\section{Data structure for partial Monge matrices}
\label{sec:partial}
Our goal in this section is to extend the solution described in Section~\ref{sec:linear} to \emph{partial}
Monge matrices. Recall that in a partial Monge matrix $M$, for any
$i<j$ and $k < \ell$, the condition $M[i,k]+M[j,\ell] \ge
M[i,\ell]+M[j,k]$ holds only if all of $M[i,k],M[j,\ell],M[i,\ell],M[j,k]$
are defined. Not all entries in $M$ are defined, but the defined
entries in every row and every column are contiguous. Let $s_{i}$ and $t_{i}$ denote the first and last columns containing defined entries in the $i$'th row respectively. 
We assume that we know the coordinates
of at least one of the defined entries. This allows us to find all $s_{i}$'s and $t_{i}$'s in $O(n\log n)$ time.

The following Lemma states that we can implicitly fill appropriate constants instead of the undefined (blank) entries to turn a partial Monge matrix into a full Monge matrix:

\begin{lemma}\label{lemma:filltheblanks}
The blank entries in an $m\times n$ partial Monge matrix $M$ can be implicitly replaced in $O(m+n)$ time so that $M$ becomes Monge and each $M_{ij}$ can be returned in $O(1)$ time.
\end{lemma}
\begin{proof}
Let $s_i$ (resp. $t_i$) denote the index of the leftmost (resp. rightmost) column that is defined in row $i$. 
Since the defined (non-blank) entries of each row and column are continuous we have that the sequence $ s_1,s_2,\ldots,s_m$ starts with a non-increasing prefix $s_1\ge s_2\ge \ldots \ge s_a$ and ends with a non-decreasing suffix $s_a\le s_{a+1}\le \ldots \le s_m$. 
Similarly, the sequence $t_1,t_2,\ldots,t_n$ starts with a non-decreasing prefix $ t_1\le t_2\le \ldots \le t_b$ and ends with a non-increasing suffix $t_b\ge t_{b+1}\ge \ldots \ge t_m$. 

We partition the blank region of $M$ into four regions: (I) entries that are above and to the left of $M[i, s_i]$ for $i=1,\ldots,a$, 
(II) entries that are below and to the left of $M[i, s_i]$ for $i=a+1,\ldots,m$, 
(III) entries that are above and to the right of $M[i ,t_i]$ for $i=1,\ldots,b$, 
(IV) entries that are below and to the right of $M[i ,t_i]$ for $i=b+1,\ldots,n$.
We first describe how to replace all entries in region I to make them non-blank and
obtain a valid partial Monge matrix (whose blank entries are only in regions II, III, and IV). The remaining regions are handled in a similar manner, one after the other.

We describe our method for filling in the blank entries in region I in two steps. In the first step we show how to implicitly fill in the blanks in a lower right triangular Monge matrix so that each filled blank entry can be computed in $O(1)$ time. By a lower right triangular Monge matrix we mean a partial Monge square matrix with $n$ rows and columns, such that, for all $1\leq i\leq n$, $s_i = n-i+1$. In the second step we explain that any $m$-by-$n$ partial Monge matrix whose blank entries are in region I can be turned into a lower right triangular Monge matrix with at most $m+n$ rows and columns. The only operations used in the transformation are duplicating rows, duplicating columns, and turning elements into blanks. We will show an $O(m+n)$ procedure for computing two tables. One specifying, for each row index $1\leq i\leq m$, the corresponding row index in the larger $O(m+n)$ triangular matrix. The second is an analogous table for the columns indices. The lemma then follows for the blank entries in region I. The other regions are treated by reducing to the region I case, one after the other.

We now describe how to fill in the blank regions in a lower right triangular Monge matrix. Let $W$ denote the largest absolute value of any non-blank entry in $M$ (We can find $W$ by applying the
algorithm of Klawe and Kleitman~\cite{KK89}). 
Intuitively, we would like to make every $M[i,j]$ in the upper left triangle very large. However,
we cannot simply assign the same large value to all of them, because then the Monge
inequality would not be guaranteed to hold if more than one of the four considered elements
belongs to the replaced part of the matrix. 
A closer look at all possible cases shows that setting all the entries of each diagonal to the same value does work. 
More precisely, we replace the blank element $M[i,j]$ with $3W[2(n-i-j)+1]$. Thus, each element in the first diagonal off the main diagonal ($i+j=n$) is set to $3W$, the elements of the second diagonal off the main diagonal are set to $9W$, etc. Note that the maximum element in the resulting matrix is $O(nW)$. 
To prove that the resulting new matrix $M'$ is Monge, it suffices, by Proposition~\ref{prop:adjacent condition}, to show that, for all $1\leq i,k < n$,  $M'[i,k]+M'[i+1,k+1]-M'[i,k+1]-M'[i+1,k] \geq 0$.
To this end we consider the following cases:
\begin{enumerate}
\item $i+k>n$, so all $M[i,k],M[i+1,k+1],M[i,k+1],M[i+1,k]$ are non-blank, and the inequality holds because $M$
is partial Monge.
\item $i+k=n$, so $M[i,k]$ is blank and $M[i+1,k+1],M[i,k+1],M[i+1,k]$ are non-blank. Then\\
$M'[i,k]+M'[i+1,k+1]-M'[i,k+1]-M'[i+1,k] = 3W + M'[i+1,k+1]-M'[i,k+1]-M'[i+1,k] \geq 3W-3W = 0$.
\item $i+k=n-1$, so $M[i,k],M[i,k+1],M[i+1,k]$ are blank, and $M[i+1,k+1]$ is non-blank. Then\\
$M'[i,k]+M'[i+1,k+1]-M'[i,k+1]-M'[i+1,k] = 9W + M[i+1,k+1] - 3W-3W \geq 3W - W \geq 0$.
\item $i+k<n-1$, so all $M[i,k],M[i+1,k+1],M[i,k+1],M[i+1,k]$ are blank. Then,\\ 
$M'[i,k]+M'[i+1,k+1]-M'[i,k+1]-M'[i+1,k] = 0$.
\end{enumerate}
Hence the new matrix $M'$ is indeed Monge.

Next, we describe how to turn any $m$-by-$n$ partial Monge matrix $M$ whose blank entries are in region I into a slightly larger lower right triangular matrix $M'$. This is done by duplicating some rows or columns of $M$ and replacing by blanks a nonempty prefix in all but a single copy. Thus, each row $r$ (column $c$) of $M$ has at least one appearance in $M'$ in which no elements are replaced by blanks. We say that $r$ ($c$) is mapped to such an appearance in $M'$.   Propositions~\ref{prop:undef} and~\ref{prop:duplicate} guarantee that $M'$ is partial Monge. 
For ease of presentation we describe the process as if we actually transform the $M$ into $M'$, although in reality we only need to compute the mappings of rows and columns. 

The assumption that the blank entries are in region I implies that $s_1 \geq s_2 \geq \cdots \geq s_m=1$ and that $t_1=t_2=\cdots = t_m = n$. 
We first guarantee that the $s_i$'s are strictly decreasing. We do this by iterating through the $s_i$'s. If $s_i = s_{i-1}$, we duplicate the column $s_i$ of $M$ (i.e., shift all columns at indices greater than $i$ by one position, and insert a copy of column $s_i$ at the vacant index $s_i+1$), make $M[i-1,s_i]$ blank, and mark the column currently at index $s_i$ as a duplicate (the index of this column might change later if columns with smaller indices will be duplicated). This column duplication has the effect of increasing by 1 all $s_j$'s for $j<i$. Let $M'$ denote the matrix obtained from $M$ at the end of this process, and let $s'_i$ denote the index of the first defined entry in row $i$ of $M'$. If $M'$ does not have $m+n$ columns, we insert a sufficient number of copies of first column of $M'$ to make it so. 
We construct a table $c[\cdot]$ which keeps track of the mapping of columns of $M$ to $M'$ by recording for each non-duplicate row its original index in $M$ and its index in $M'$. Clearly, computing $c[\cdot]$ and the indices $s'_i$'s can be done in $O(m)$ time without actually duplicating the columns. See Figure~\ref{fig:stretch} (middle) for an illustration.

We will further modify $M'$ and 
 use a table $r[\cdot]$ to keep track of the mapping from rows of $M$ to rows of $M'$. 
For convenience, we define $s'_0 = m+n+1$, and $r[0] =0$. We iterate through the sequence $s'_1,s'_2, \dots, s'_m$. 
We add to $M'$ $s'_{i-1}-s'_i$ copies of row $i$ of $M'$, and, for $j=1,2, \dots, s'_{i-1}-s'_i-1$, replace the prefix of length $j$ from the $j$'th copy by blanks, so only the last copy remains unchanged. We therefore set $r[i]$ to $r[i-1]+s'_{i-1}-s'_i$. Clearly, we can compute the table $r[\cdot]$ in $O(m+n)$ time without actually modifying $M'$.
See Figure~\ref{fig:stretch} (right) for an illustration.

\begin{figure}[h!]
\centering
   \includegraphics[width=0.9\textwidth]{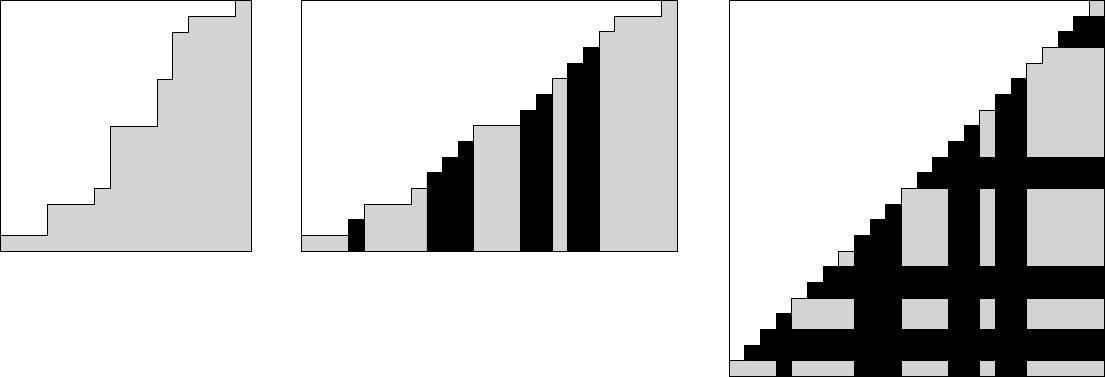}
   \caption{
   A staircase matrix $M$ is transformed into a triangular matrix $M'$ in two steps by duplicating columns and rows. Defined entries
   are gray, undefined white, and duplicated columns/rows black.\label{fig:stretch}}
 \end{figure}

Finally, to obtain the value with which the blank entry at $M[i,j]$ should be replaced when converting $M$ into a full Monge matrix, we return $3W[2(m+n-r[i]-c[j])+1]$.

Regions II, III, and IV can be handled symmetrically to region I. To handle undefined entries in region II, we 
implicitly reverse the order of the rows and negate all the elements of the matrix. It is easy to verify that the resulting matrix is Monge with undefined entries in region I. We then implicitly fill in the undefined values using the method described above, negate all the elements and revert the order of rows to its original order. The transformation for region III is reversing the order of columns and negating all elements, and the transformation for region IV is reversing the order of both rows and columns.
Note that to make $M$
full Monge we first need to fill the blanks in region I, then calculate the new value of $W$
and fill the blanks in region II accordingly, and so on.
\qed \end{proof}

For subcolumn (or subrow) maximum queries, the above lemma implies that we can handle partial Monge matrices in the same bounds as full Monge matrices (i.e., the bounds of Theorem~\ref{thm:subcolumn2} and Corollary~\ref{lem:micro} also apply to partial Monge matrices). 
Upon subcolumn query (a column $c$ and a range of rows $R$) we first restrict $R$ to the defined entries in the column  $c$ and only then query the data structure. 

For submatrix queries however, this trick only works if the query range is entirely defined. In general, it does not work because the defined entries in the query range do not necessarily form a submatrix. Handling submatrix queries is therefore more complicated. Our solution is based on the following decomposition. 

\subsection{Decomposing a partial Monge matrix into staircase matrices}  

Our data structure relies on a decomposition of $M$ into {\em staircase}
matrices. A partial matrix is staircase if the defined entries in its rows either all  begin in the first column and the $t_i$s are monotone, or all end in the last column and the $s_i$s are monotone. 
It is well known (cf.~\cite{AggarwalK90}) that by cutting
$M$ along columns and rows, it can be decomposed into staircase
matrices $\{M_i\}$ such that each row is covered by at most two matrices,
and each column is covered by at most three  matrices. 
For completeness, we  describe such a decomposition below.

\begin{lemma}\label{lemma:decomposition}
A partial matrix $M$ can be decomposed into staircase
matrices $\{M_i\}$ such that each row is covered by at most two matrices,
and each column is covered by at most three  matrices.
\end{lemma}
\begin{proof}

Let $s_i$ and $t_i$ denote the smallest and largest column index
in which an element in row $i$ is defined, respectively. 
The fact that the defined entries of $M$ are contiguous in both rows
and columns implies that the sequence $s_1, s_2, \dots, s_m$ consists of a
non-increasing prefix and a non-decreasing suffix. Similarly, the 
sequence $t_1, t_2, \dots, t_m$ consists of a
non-decreasing prefix and a non-increasing suffix. 
It follows that the rows of $M$ can be divided into three ranges - 
a prefix where $s$ is non-increasing and $t$ is non-decreasing, an infix where
both $s$ and $t$ have the same monotonicity property, and a suffix
where $s$ is non-decreasing and $t$ is non-increasing.
The defined entries in the prefix of the rows can be divided into two
staircase matrices by splitting $M$ at $t_1$, the largest column where the
first row has a defined entry. 
Similarly, the defined entries in the suffix of the rows can be divided into two
staircase matrices by splitting it at $t_m$, the largest column where the
last row has a defined entry. 
The defined entries in the infix of the rows form a double staircase
matrix. It can be broken into staircase matrices by dividing along
alternating rows and columns as shown in Figure~\ref{fig:partial}. 

\begin{figure}[h!]
\centering
   \includegraphics[scale=0.6]{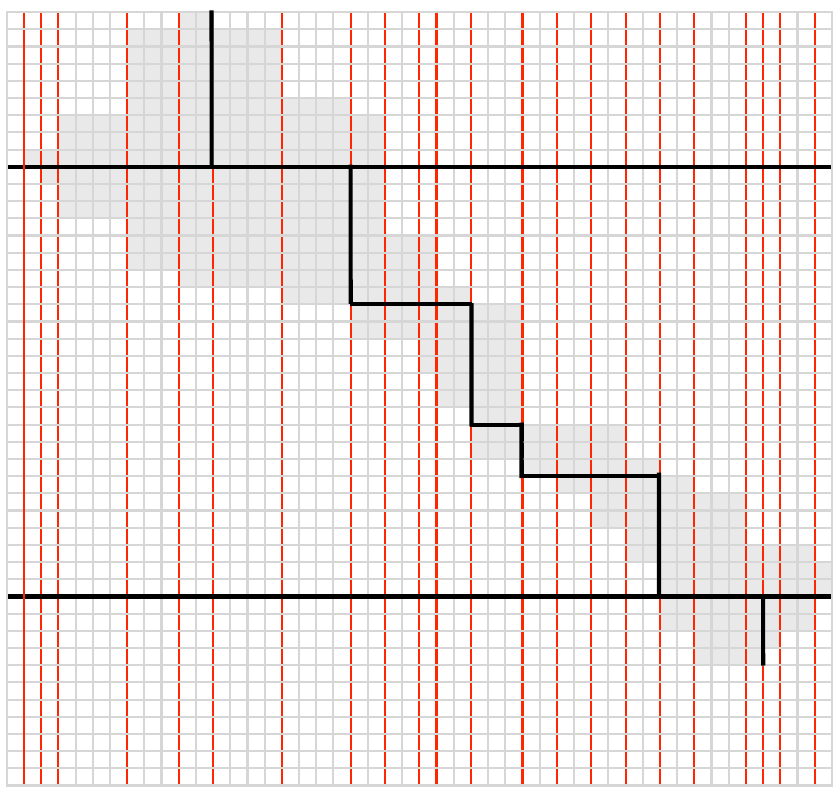}
   \caption{
   A decomposition of a partial matrix (where the defined entries are gray and the undefined white) into staircase matrices
     (defined by solid thick black lines) and into blocks of consecutive
     columns with the same defined entries (indicated by thin vertical
     red lines).\label{fig:partial}}
 \end{figure}

It is easy to verify that, in the resulting decomposition, each row
is covered by at most two staircase matrices, and each column is covered
by at most three staircase matrices.
Also note that  every set of consecutive
columns whose defined elements are in exactly the same set of rows are covered
in this decomposition by the same three row-disjoint staircase matrices. \qed
\end{proof}

Before we use the above decomposition for our data structure, we show how it can be used  to prove that, if $M$ is an $m\times n$ TM (or Monge) staircase
matrix, then the number of breakpoints of $M$ is
$O(m)$. This result illustrates the use of the decomposition, it was used in the data structure of~\cite{ourICALP}, and we believe is of independent interest.

\begin{theorem}\label{lemma:partial-breakpoints}
Let $M$ be a partial $m\times n$ matrix in which the defined entries in each row
and in each column are
contiguous.
If $M$ is TM (i.e., for all $i<j, k<\ell$ where
$M[i,k],M[i,\ell],M[j,k],M[j,\ell]$ are all defined, $M[i,k] \leq M[j,k]
\implies M[i,\ell] \leq M[j,\ell]$), then the number of breakpoints of $M$ is $O(m)$.
\end{theorem}
\begin{proof}
We first show that the number of breakpoints of an $m\times n$ TM {\em staircase} matrix is at
most $2m$.
We focus on the case where the defined entries of all rows begin in the first column and end in non-decreasing columns. In other words, for all $i$, $s_i$=1 and $t_i\le t_{i+1}$. The other cases are symmetric. 

A breakpoint is a situation where the maximum in column $c$ is at row
$r_1$ and the maximum in column $c+1$ is at a different row $r_2$.
We say that $r_1$ is the departure row of the breakpoint, and $r_2$ is
the  entry row of the breakpoint.
There are two types of breakpoints: decreasing ($r_1 < r_2$), and
increasing  ($r_1 > r_2$).
We show that 
(1) each row can be the entry row of at most one decreasing breakpoint, and (2) each row can be the
departure row of at most one increasing breakpoint. 
\begin{enumerate}
\item[(1)] Assume that row $r_2$ is an entry row of two decreasing
  breakpoints:  One is the pair of entries $(r_1,c_1),(r_2,c_1+1)$ and
  the other is the pair  $(r_3,c_2),(r_2,c_2+1)$. We know that
  $r_1<r_2$, $r_3<r_2$, and wlog $c_2>c_1+1$. 
Since the maximum in column $c_1+1$ is in row $r_2$, we have
$M[r_3,c_1+1] < M[r_2,c_1+1]$.
However, since the maximum in column $c_2$ is in row $r_3$, we have
$M[r_3,c_2]  > M[r_2,c_2]$, contradicting the total monotonicity of
$M$. Note that $M[r_2,c_2]$ is defined since $M[r_2,c_2+1]$ is defined.

\item[(2)]
Assume that row $r_1$ is a departure row of two increasing breakpoints:
One is the pair of entries $(r_1,c_1),(r_2,c_1+1)$ and the other is
the pair  $(r_1,c_2),(r_3,c_2+1)$. We know that $r_1>r_2$ and
$r_1>r_3$. 
Since the maximum in column $c_1$ is in row $r_1$, we have
$M[r_2,c_1] < M[r_1,c_1]$.
However, since the maximum in column $c_1+1$ is in row $r_2$, we have
$M[r_2,c_1+1]  > M[r_1,c_1+1]$, contradicting the total monotonicity of
$M$. Note that $M[r_1,c_1+1]$ is defined since $M[r_1,c_2]$ is defined.
\end{enumerate} 

The above two claims prove that the number of breakpoints of a staircase matrix is at
most $2m$.
We use this fact, and the above decomposition to staircase matrices to prove an $O(m)$ bound for arbitrary partial matrices.

Let $bp(M_i)$ denote the number of breakpoints in matrix $M_i$. 
Let $m_i$ denote the number of rows in $M_i$.
Since each row appears in at most two $M_i$s, $\sum_i m_i =
O(m)$.
The total number of breakpoints in all  $M_i$s is
$O(m)$ since 
$\sum_i bp(M_i) = \sum_i O(m_i) = O(m)$.

Consider now a partition of $M$ into rectangular blocks $B_j$ defined by maximal
sets of contiguous columns whose defined entries are at the same set
of rows, see Figure~\ref{fig:partial}. There are $O(m)$ such blocks.
Notice that the number of breakpoints of $M$ is $bp(M) = \sum_j bp(B_j) + O(m)$ (the
$O(m)$ term accounts for the possibility of a new breakpoint between every two
consecutive blocks). Therefore, it suffices to bound $\sum_j bp(B_j)$.

Consider some block $B_j$. As we mentioned above, the columns of $B_j$
appear in at most three row-disjoint staircase matrices $M_1,M_2,M_3,$ in the decomposition of
$M$. The column maxima of $B_j$ are a subset of the column maxima of
$M_1,M_2,M_3$. Assume wlog that the indices of rows covered by $M_i$ are smaller than
those covered by $M_{i+1}$ for every $i=1,2$. 

The breakpoints of $B_j$
are either breakpoints of $M_1,M_2,M_3$, or
breakpoints that occur when the maxima in consecutive columns of $B_j$
originate in different $M_i$. However, since $B_j$ is a (non-partial)  TM matrix, its column maxima are
monotone. So once a column maximum originates in $M_i$, no maximum in
greater columns will ever originate in $M_j$ for $j<i$. It follows
that the number of breakpoints in $B_j$ that are not breakpoints of
$M_1,M_2,M_3$ is at most two. Since there are  $O(m)$ blocks, 
$\sum_j bp(B_j) \leq \sum_i bp(M_i) + O(m) = O(m)$. This completes the proof of Theorem~\ref{lemma:partial-breakpoints}.\qed \end{proof}

\subsection{The data structure}

We begin with a weaker result (Theorem~\ref{thm:staircasesubmatrixlarge}), which is that one can answer submatrix maximum queries on  an $n \times n$ staircase matrix in $O(\log\log n)$ time
with a structure of size $O(n \log n)$. We will then (Theorem~\ref{thm:staircasesubmatrix}) show how to reduce the space to $O(n)$,  and finally (Theorem~\ref{thm:partialsubmatrix}) how to handle arbitrary partial Monge matrices using the decomposition into staircase matrices.

We will need the following preliminary lemma, that follows quite easily from the persistent predecessor structure of Chan~\cite{ChanPersistent}. 

\begin{lemma}
\label{lem:dominancemaximum}
A collection $S$ of $O(n)$ weighted points on an $n\times n$ grid can be preprocessed in $O(n\log\log n)$ time and $O(n)$ space,
so that, given any $(x,y)$, the maximum weight of a point $(x',y')\in S$ such that $x'\geq x$ and $y'\geq y$ can be calculated in
$O(\log\log n)$ time.
\end{lemma}

\begin{proof}
We use the standard geometric idea of sweeping the grid with a horizontal line while maintaining a data structure describing the current
situation. The data structure is made partially persistent so that after sweeping, given a query $(x,y)$, we can retrieve the version
of the structure corresponding to a horizontal line passing through $(x,y)$. Querying that version of the data structure will allow us
to answer the request. The data structure will be a predecessor structure made persistent using the result of Chan~\cite{ChanPersistent}.
See Theorem 5 of~\cite{MosheSurvey} 
for a more detailed description of a similar lemma.

Denote the points by $(x_{i},y_{i})$ and their corresponding weights by $w_{i}$. We assume that the weights are distinct.
We sweep the grid with a horizontal line starting at $y=n$. The predecessor structure stores $x$-coordinates of some of the already
seen points. Coordinate $x_{i}$ is stored in the predecessor structure iff $y_{i}\geq y$ and there is no $i'$
such that $y_{i'}\geq y$, $x_{i'}\geq x_{i}$ and $w_{i'}>w_{i}$. This is because otherwise the $i'$-th point is a better answer
than the $i$-th point for any query processed using this or any future version of the data structure.
Consequently, the points whose $x$-coordinates are stored in the predecessor structure can be arranged so that
their $x$-coordinates are increasing and the weights decreasing. Then it follows that locating the maximum weight
of a point $(x',y')\in S$ such that $x' \geq x$ and $y' \geq y$ can be done by finding the successor of $x$ in the
version of the predecessor structure corresponding to $y$. Maintaining the structure while sweeping the grid
is also done with a predecessor search. After having seen a new point $(x_{i},y_{i})$ we locate the predecessor of $x_{i}$.
If the weight of the corresponding point is smaller than $w_{i}$, we remove it from the structure and repeat.

A persistent predecessor search structure can be implemented in space $O(n)$ while keeping the query
time $O(\log\log n)$~\cite{ChanPersistent}. Consequently, we can build in $O(n\log\log n)$ time a~structure of
size $O(n)$ answering queries in $O(\log\log n)$ time. \qed
\end{proof}

\begin{theorem}
\label{thm:staircasesubmatrixlarge}
Given an $n\times n$ staircase Monge matrix $M$, a data structure of size $O(n \log n)$ can be constructed
in $O(n \log n)$ time to answer submatrix maximum queries in $O(\log\log n)$ time.
\end{theorem}
\begin{proof}
Because of left-right symmetry, we can assume that the defined entries in row $i$ 
start in the first column and end in column $t_{i}$. Notice that either $t_{1}\le t_{2} \le \ldots \le t_{n}$ or $t_{1}\ge t_{2} \ge \ldots \ge t_{n}$.
Without loss of generality we will assume the latter. This is enough because we will not be explicitly using the Monge property in our
solution, except for applying Theorem~\ref{thm:submatrix2} on a copy of $M$ (called $\widetilde M$) where the undefined entries are
appropriately filled. 
 
We partition $M$ into full Monge matrices using a standard method: First, create a full Monge matrix by taking the upper-left fragment $[1,n/2]\times [1,t_{n/2}]$ of $M$. Then, recursively decompose the staircase  matrices created by taking the upper-right fragment $[1,n/2]\times [t_{n/2}+1,n]$ and the lower-left fragment $[n/2+1,n]\times [1,n]$ of $M$. See Figure~\ref{fig:staircaseDecoposition}. It is easy to verify that
the decomposition consists of at most $2n$ full Monge matrices (called fragments). The decomposition has other useful
properties on which we elaborate further.

\begin{figure}[h]
\begin{center}
\includegraphics[width=1\textwidth]{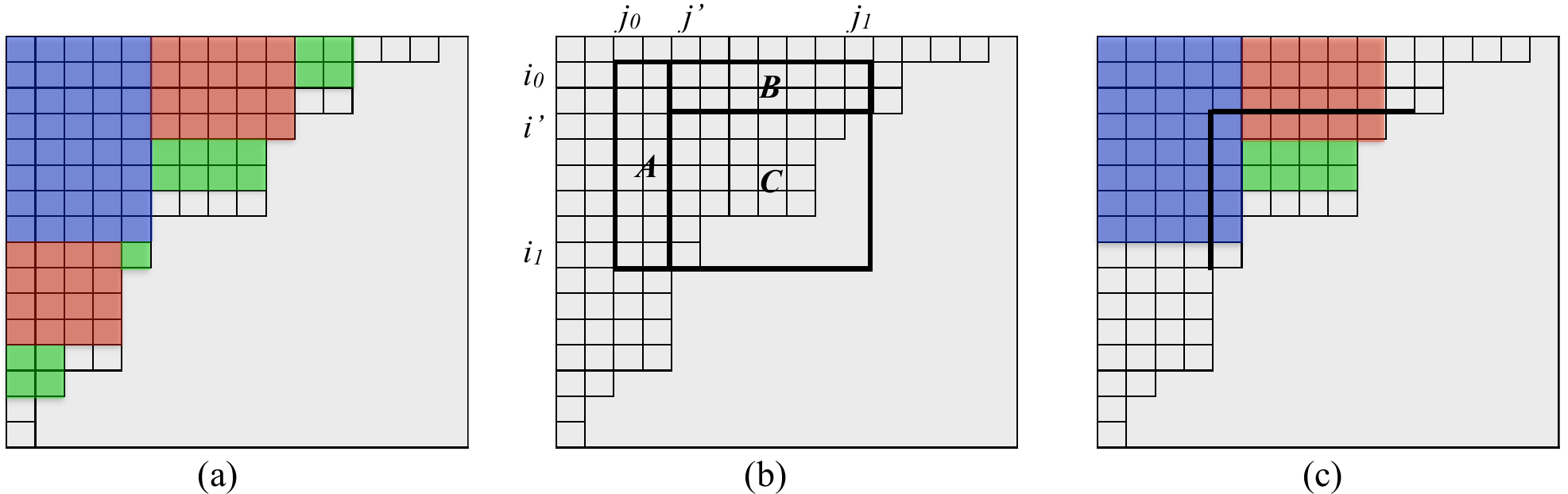}
\end{center}
\caption{(a) A staircase $n\times n$ Monge matrix partitioned into $2n$ smaller full Monge matrices (fragments). (b) A query range $[i_{0},i_{1}]\times  [j_{0},j_{1}]$ decomposed into two full Monge matrices $A$ and $B$ and one dominance query $C$. (c) The dominance query as vertical and horizontal lines (the green fragment is fully inside the range and the blue and red fragment intersect the horizontal line). }
\label{fig:staircaseDecoposition}
\end{figure}

Consider a query range $[i_{0},i_{1}]\times  [j_{0},j_{1}]$.
To find the maximum (defined) $M[i,j]$ over all $i\in [i_{0},i_{1}]$ and $j\in [j_{0},j_{1}]$ 
we proceed as follows. The simple case is when the query range is fully within the defined part of $M$. To handle this case, we apply Theorem~\ref{thm:submatrix2} on a copy of $M$ (denoted $\widetilde M$) where the undefined entries are appropriately (and implicitly) filled using Lemma~\ref{lemma:filltheblanks}. This allows us to do submatrix queries in $O(\log\log n)$ time when the query range is fully defined.
Otherwise, we decompose the query into three parts. 
The first part, which we call a \emph{dominance maximum query}, is to
find the maximum $M[i,j]$ over all $i \geq i'$ and $j\geq j'$, for
$i',j'$ to be defined shortly.
The other two are submatrix maximum queries fully within the defined part of $M$ (and hence can be processed by querying the
structure built for $\widetilde M$ in $O(\log\log n)$ time). The decomposition is performed in $O(1)$ time by setting
$j'=t_{i_{1}+1} +1$ and choosing the smallest $i'\geq i_{0}$ such that $t_{i'} < j_{1}$ (which can be preprocessed for every possible
$j_{1}$ in $O(n)$ space). The two submatrix maximum queries are therefore over the full Monge matrices $[i_{0},i_{1}]\times [j_{0},j'-1]$ and  $[i_{0},i'-1]\times [j',j_{1}]$. Hence,
it is enough to focus on answering dominance maximum queries.

To answer a dominance maximum query (i.e., to find the maximum $M[i,j]$ over all $i \geq i'$ and $j\geq j'$) we use the partition of $M$ into full Monge matrices (fragments). Every such fragment
is either fully inside the query range, fully outside of the query range, or intersected by the query range boundary. 

\paragraph{\bf Fragments inside the query range.}
A fragment $[r_{0},r_{1}]\times [c_{0},c_{1}]$ is fully inside the query range iff
$r_{0} \geq i'$ and $c_{0} \geq j'$. This observation allows us to reduce computing the maximum
over all matrices fully inside the query to the problem defined in Lemma~\ref{lem:dominancemaximum}.
The reduction is simply that for every fragment $[r_{0},r_{1}]\times [c_{0},c_{1}]$ we create
a point $(r_{0},c_{0})$ and set its weight to be the maximum inside the fragment. As a result, we
create at most $O(n)$ points on the $n\times n$ grid. 
Using Theorem~\ref{thm:submatrix} on $\widetilde M$ to create every point separately 
takes total $O(n\log\log n)$ in the preprocessing time, so in $O(n\log\log n)$ time
we can construct a structure of size $O(n)$ answering queries in $O(\log\log n)$ time.

\paragraph{\bf Fragments intersected by the query range.}
We are left only with finding the maximum over all fragments intersected by the boundary of our dominance
maximum query. We partition these fragments into three groups. The first consists of the single fragment
containing $M[i',j']$. The maximum there can be found with a submatrix maximum query
on $\widetilde{M}$ in $O(\log\log n)$ time.
All other fragments intersected by the boundary are either intersected by the horizontal line $y=i'$
or the vertical line $x=j'$, but not both. We show how to find the maximum over all matrices intersected by the horizontal
line $y=i'$ and fully to the right of the vertical line $x=j'$ (the other case is symmetric). 

By the properties of our decomposition
scheme, there are at most $\log n$ fragments intersected by any horizontal line, and
they can be arranged in the natural left-to-right order. For every possible horizontal line,
we store these at most $\log n$ fragments in an array. For every fragment we store the coordinates
of its corresponding submatrix of $M$ and the maximum in all of its entries below the horizontal
line. The array is additionally equipped with the maximum over all maxima in each one of its suffixes.
Such preprocessed data allows us to find the maximum over all fragments intersected by
a horizontal line $y=i'$ and fully on the right of a vertical line $x=j'$ in $O(\log\log n)$ time:
First, we binary search over the array stored for $y=i'$ to locate the leftmost fragment completely on
the right of $x=j'$. Then we return the stored corresponding maximum. Notice that
the binary search also allow us to locate the fragment containing $M[i',j']$. Consequently,
the whole query time is $O(\log\log n)$ using $O(n\log n)$ space for this part
of the implementation. To guarantee $O(n\log n)$ preprocessing time, we  run the 
SMAWK algorithm on every fragment in the decomposition in total $O(n\log n)$ time. This gives us the maximum in every row of every fragment. This is then enough to construct all arrays in $O(n\log n)$ time.
\qed \end{proof}

We now proceed to improving Theorem~\ref{thm:staircasesubmatrixlarge} so that the structure needs just linear space. The main idea
is to partition the $n\times n$ staircase matrix $M$ into cells of size $\log n \times \log n$ and then define a new smaller $(n/\log n) \times (n/\log n)$
staircase matrix $M'$ (whose entries correspond to cell-maxima in $M$) on which we apply Theorem~\ref{thm:staircasesubmatrixlarge}. To implement this idea
we need a number of additional auxiliary data structures, which take $O(n)$ space in total. 
We start with an auxiliary lemma, which will be used to provide constant-time access to entries of $M'$. 

\begin{lemma}
\label{lem:cellmaximum}
Given an $n\times n$ Monge matrix $M$ partitioned into $\log n \times \log n$ cells, a data structure
of size $O(n)$ can be constructed in $O(n\log n)$ time to find the maximum in a given cell in $O(1)$ time.
\end{lemma}

\begin{proof}
We partition $M$ into $n/\log n$ horizontal slices, each consisting of $\log n$ rows (and all columns). Consider a single slice,
which is a $\log n \times n $ Monge matrix.
We store its breakpoints $c_{1}<c_{2}<\ldots < c_{k}$ (where $k \le \log n$) in an atomic heap,
consequently allowing predecessor queries in $O(1)$ time (this is exactly how the structure from
Corollary~\ref{lem:micro} works). Additionally, similarly to Lemma~\ref{lem:submatrix}, for every $i\geq 2$
we precompute the value of
$$m_{i}=\max_{j\in [c_{i-1},c_{i})} M[r(c_{i-1}),j]$$
and augment these values with a (one dimensional) range maximum data structure. Here, $r(c)$ denotes the row containing
the maximum element in the $c$-th column of the slice in question. 
Using two predecessor queries
and one range maximum query, the problem of finding the maximum in a given cell
(which is fully contained in a single horizontal slice) reduces in $O(1)$ time to finding the maximum in at most two rows.
The total space is $O(n/\log n \cdot \log n)=O(n)$ and the bottleneck
in the preprocessing is computing the breakpoints for all slices.
The breakpoints of a single slice can be computed in $O(\log ^2 n)$ by
adding one row at a time, as done in the proof of
Lemma~\ref{lem:subcolumn}. In total, this  takes $O(n/\log n \cdot \log^2
n)=O(n\log n)$ total time. 

We repeat the above reasoning on the transpose of $M$. As a result, we either already know the maximum
element, or we have isolated at most two rows and at most two columns, such that the maximum lies in one of these rows and one of these columns. This gives us at most four candidates
for the maximum, which can be retrieved and compared naively. \qed
\end{proof}

\noindent We are now ready to present our linear-space improvement to  Theorem~\ref{thm:staircasesubmatrixlarge}.

\begin{theorem}
\label{thm:staircasesubmatrix}
Given an $n\times n$ staircase Monge matrix $M$, a data structure of size $O(n)$ can be constructed
in $O(n\log n)$ time to answer submatrix maximum queries in $O(\log\log n)$ time.
\end{theorem}

\begin{proof}
As in the proof of Theorem~\ref{thm:staircasesubmatrixlarge}, we can assume that the defined entries in row $i$ 
start in the first column and end in column $t_{i}$, and that  $t_{1}\ge t_{2} \ge \ldots \ge t_{n}$.

We partition $M$ into cells of size $\log n \times \log n$ and then define a  smaller  $(n/\log n) \times (n/\log n)$
staircase matrix $M'$. Notice that, unlike Lemma~\ref{lem:cellmaximum}, $M$ is a staircase Monge matrix (and not a full Monge matrix). This means that there are three types of cells in $M$: fully defined, partially defined, and fully undefined. An entry of $M'$ is defined iff its corresponding cell in $M$ is
fully defined. In this case the entry is equal to the maximum in the corresponding
cell. The undefined entries of $M'$ are the ones corresponding to either partially defined or fully undefined cells of $M$. 
We appropriately (and implicitly)  fill these entries using Lemma~\ref{lemma:filltheblanks} to turn $M$' into a full Monge
matrix $\widetilde M'$, on which we apply Lemma~\ref{lem:cellmaximum}. This gives us
constant-time access to the entries of $M'$, so finally we can apply Theorem~\ref{thm:staircasesubmatrixlarge}
to preprocess it in $O(n)$ space and $O(n\log n)$ time to answer submatrix maximum queries
in $O(\log\log n)$ time.

Regarding partially defined
cells, we observe that there are at most $2n/\log n$ of them. Furthermore, they can be arranged in a linear order, so that if the part of $M$
corresponding to the $i$-th partially defined cell is $[r_{i},r'_{i}] \times [c_{i},c'_{i}]$, then
for all $i$ either $[r_{i},r'_{i}]=[r_{i+1},r'_{i+1}]$ and $c'_{i}+1=c_{i+1}$ or $r_{i}=r'_{i+1}+1$ and $[c_{i},c'_{i}]=[c_{i+1},c'_{i+1}]$
(to be more precise, we might need to declare some fully defined cells partially defined to
guarantee this property). We create a predecessor structure storing all $r_{i}$s and a separate
predecessor structure storing all $c_{i}$s. We also compute the maximum in every partially
defined cell and store them in an array (arranged in the aforementioned linear order) augmented
with a (one dimensional) range maximum structure. Computing the maximum in all partially defined
cells is done in $O(n/\log n \cdot \log n \cdot \alpha(\log n))=O(n\cdot \alpha(\log n))$ time using~\cite{KK89}.

By the same reasoning given in the proof of Theorem~\ref{thm:staircasesubmatrixlarge}, it is enough
to implement dominance maximum queries on $M$. A dominance maximum query can be
decomposed into (i) a dominance maximum query on $\widetilde M'$, which can be answered
in $O(\log\log n)$ time, (ii) finding the maximum inside all partially defined cells fully within
the query range, and (iii) finding the maximum inside partially defined cells intersected by the
boundaries of the query range. All partially defined cells fully within the query range create
a contiguous interval in the linear order. The range can be determined in $O(\log\log n)$
using the predecessor structures storing all $r_{i}$s and $c_{i}$s, and then the maximum can
be found in $O(1)$ time with a (one dimensional) range maximum query. It remains to calculate the maximum inside partially defined cells
intersected by the boundaries of the query range. We will describe how to process all partially
defined cells intersected by the horizontal boundary. Handling the vertical boundary is symmetric.

Let the dominance maximum query be specified by $(i',j')$. We want to compute the maximum
inside the query range and belonging to a partially defined cell intersected by the horizontal line $y=i'$.
All such cells create a contiguous interval in the linear order, which can be determined
with two predecessor queries in $O(\log\log n)$ time. In the same complexity, we can find
the leftmost such cell $u$ which is not fully on the left of the vertical line $x=j'$. We decompose
the original query into a dominance maximum query inside $u$, and the remaining
part. The remaining part starts at a left boundary of a partially defined cell and consists
of the entries at or below $y=i'$ in all partially defined cells to
the right of $u$.
Consequently, the answer can be preprocessed
for every point on a left boundary of a partially defined cell using $O(n/\log n \cdot \log n)=O(n)$
space and $O(n/\log n \cdot \log n \cdot \alpha(\log n))=O(n\cdot \alpha(\log n))$ time using~\cite{KK89}.
The bottleneck in the preprocessing is computing the maximum in every row of every
partially defined cell.

It remains to describe how to handle the dominance query in $u$. In
other words, after constructing in $O(n\log n)$ time an $O(n)$ size
structure, we have,  in $O(\log\log n)$ time, 
reduced an arbitrary dominance maximum query into a dominance maximum query inside
a single partially defined cell. This cell is a smaller $\log n \times \log n$ staircase
 matrix, and furthermore there are at most $2 n/\log n$ such cells. By recursing on each of
 these smaller staircase matrices separately, we construct in
 additional $O(n/\log n \cdot \log n \log\log n) = O(n \log\log n)$ time an
 $O(n/\log n \cdot \log n) = O(n)$ size structure, which reduces the
 original dominance query, in
 additional $O(\log\log\log n)$ time, into 
 a dominance maximum query inside one of $O(n/\log n \cdot \log n/\log\log n)=O(n/\log\log n)$ tiny $\log\log n \times \log\log n$
 staircase matrix (each of them being a submatrix of the original $M$). By recursing again
 on every tiny staircase matrix separately,
 we construct in additional $O(n\log \log \log n)$ time an $O(n)$ size
 structure, which reduces the original arbitrary dominance query in
 additional $O(\log\log\log\log n)$  time  into a dominance maximum query inside 
 an $(\log\log\log n)\times (\log\log\log n)$ submatrix of $M$. Such dominance maximum
 query can be answered naively resulting in $O(\log\log n+(\log\log\log n)^{2})=O(\log\log n)$ total query time.
 \qed \end{proof}

We are now ready to prove the main theorem of this section, which is that using Theorem~\ref{thm:staircasesubmatrix} we can actually implement submatrix maximum queries on arbitrary (and not just staircase) partial Monge
matrices. The idea is to partition the partial Monge matrix into
staircase matrices, so that each row and each column belong to $O(1)$
staircase matrices.  Such partitioning was used in~\cite{AggarwalK90,ourICALP} . We build the data structure of
Theorem~\ref{thm:staircasesubmatrix} on each staircase matrix in the
decomposition, and build an additional data structure for queries
spanning more than one staircase matrix.

\begin{theorem}
\label{thm:partialsubmatrix}
Given an $n\times n$ partial Monge matrix $M$, a data structure of size $O(n)$ can be constructed
in $O(n\log n)$ time to answer submatrix maximum queries in $O(\log\log n)$ time.
\end{theorem}

\begin{proof}
We partition $M$ into staircase matrices as done in the proof of Lemma~\ref{lemma:decomposition} (depicted in
Figure~\ref{fig:partial}). 
Recall that the partition divides the rows of $M$ into three ranges. The first range
contributes two staircase matrices, the second range creates a double staircase
matrix, which is further broken into multiple staircase matrices, and the third range
contributes two staircase matrices. 
It is easy to verify that, in the resulting decomposition, each row
is covered by at most two staircase matrices, and each column is covered
by at most three staircase matrices.
Additionally, the staircase matrices contributed by the second range
can be partitioned into two \emph{collections},
such that any two matrices in the same collection are row-disjoint and
column-disjoint.

The data structure consists of the following components. We apply Theorem~\ref{thm:staircasesubmatrix} on every staircase  matrix in our partition. We also store additional
data for both collections. By left-right symmetry, we can assume that
the ranges of rows and columns of the matrices in the collection
are $[r_{1},r_{2}), [r_{2},r_{3}),\ldots$ and $[c_{1},c_{2}),[c_{2},c_{3}),\ldots$, respectively.
We create a predecessor structure storing all $r_{i}$'s and a separate predecessor structure
storing all $c_{i}$'s. We also compute and store the maximum inside every staircase 
matrix in the collection (this is done in total $O(n\cdot \alpha(n))$ time using the algorithm of Klawe and Kleitman~\cite{KK89}), and augment these maxima with a (one dimensional) range maximum structure.

Now consider a submatrix maximum query  $[i_{0},i_{1}]\times [j_{0},j_{1}]$. We first query the $O(1)$ structures built for
the staircase matrices corresponding to the first and the third range.
Next, we consider each of the two collections separately. To find the maximum $M[i,j]$ over
all $i\in [i_{0},i_{1}]$ and $j\in [j_{0},j_{1}]$, we use the predecessor structures
to determine in $O(\log \log n)$ the following values (without loss of generality,
they all exist):
\begin{enumerate}
\item $i'_{0}$ such that $i_{0}\in [r_{i'_{0}},r_{i'_{0}+1})$,
\item $i'_{1}$ such that $i_{1}\in [r_{i'_{1}},r_{i'_{1}+1})$,
\item $j'_{0}$ such that $j_{0}\in [c_{j'_{0}},c_{j'_{0}+1})$,
\item $j'_{1}$ such that $j_{1}\in [c_{j'_{1}},c_{j'_{1}+1})$.
\end{enumerate}
We then query the structures built for the $(i'_{0})$-th, $(i'_{1})$-th,
$(j'_{0})$-th, and $(j'_{1})$-th staircase matrix in the collection.
Now either we have already found the maximum, or it belongs
to one of the staircase matrices fully contained in the query range. Consequently, the maximum
can be found in $O(1)$ time with a single (one dimensional) range maximum query.
\qed \end{proof}

\bibliographystyle{plain}

\begin{thebibliography}{10}

\bibitem{AggarwalK90}
A.~Aggarwal and M.~Klawe.
\newblock Applications of generalized matrix searching to geometric algorithms.
\newblock {\em Discrete Appl. Math.}, 27:3--23, 1990.

\bibitem{SMAWK}
A.~Aggarwal, M.~M. Klawe, S.~Moran, P.~Shor, and R.~Wilber.
\newblock Geometric applications of a matrix-searching algorithm.
\newblock {\em Algorithmica}, 2(1):195--208, 1987.

\bibitem{AlstrupEtAl}
S.~Alstrup, G.~S. Brodal, and T.~Rauhe.
\newblock New data structures for orthogonal range searching.
\newblock In {\em 41st FOCS}, pages 198--207, 2000.

\bibitem{AFL07}
A.~Amir, J.~Fischer, and M.~Lewenstein.
\newblock Two-dimensional range minimum queries.
\newblock In {\em 18th CPM}, pages 286--294, 2007.

\bibitem{BrodalESA}
G.~S. Brodal, P.~Davoodi, M.~Lewenstein, R.~Raman, and S.~S. Rao.
\newblock Two dimensional range minimum queries and {F}ibonacci lattices.
\newblock In {\em 20th ESA}, pages 217--228, 2012.

\bibitem{BrodalDR10}
G.~S. Brodal, P.~Davoodi, and S.~S. Rao.
\newblock On space efficient two dimensional range minimum data structures.
\newblock In {\em 18th ESA}, pages 171--182, 2010.

\bibitem{BKR96}
R.~E. Burkard, B.~Klinz, and R.~Rudolf.
\newblock Perspectives of {M}onge properties in optimization.
\newblock {\em Discrete Appl. Math.}, 70:95--161, 1996.

\bibitem{ChanPersistent}
T.~M. Chan.
\newblock Persistent predecessor search and orthogonal point location on the
  word {RAM}.
\newblock {\em ACM Trans. Algorithms}, 9(3):22:1--22:22, 2013.

\bibitem{Patrascu}
T.~M. Chan, K.~G. Larsen, and M.~P{\v a}tra{\c s}cu.
\newblock Orthogonal range searching on the {RAM}, revisited.
\newblock In {\em 27th SOCG}, pages 354--363, 2011.

\bibitem{Chazelle88}
B.~Chazelle.
\newblock A functional approach to data structures and its use in
  multidimensional searching.
\newblock {\em SIAM Journal on Computing}, 17:427--462, 1988.

\bibitem{CR1989}
B.~Chazelle and B.~Rosenberg.
\newblock Computing partial sums in multidimensional arrays.
\newblock In {\em 5th SOCG}, pages 131--139, 1989.

\bibitem{DemaineLandauWeimann}
E.~D. Demaine, G.~M. Landau, and O.~Weimann.
\newblock On {C}artesian trees and range minimum queries.
\newblock {\em Algorithmica}, 68(3):610--625, 2014.

\bibitem{Munro}
A.~Farzan, J.~I. Munro, and R.~Raman.
\newblock Succinct indices for range queries with applications to orthogonal
  range maxima.
\newblock In {\em 39th ICALP}, pages 327--338, 2012.

\bibitem{FredmanW94}
M.L. Fredman and D.E. Willard.
\newblock Trans-dichotomous algorithms for minimum spanning trees and shortest
  paths.
\newblock {\em J. Comput. Syst. Sci.}, 48(3):533--551, 1994.

\bibitem{GBT84}
H.~Gabow, J.~L. Bentley, and R.E Tarjan.
\newblock Scaling and related techniques for geometry problems.
\newblock In {\em 16th STOC}, pages 135--143, 1984.

\bibitem{GawrychowskiLN14}
P.~Gawrychowski, M.~Lewenstein, and P.~K. Nicholson.
\newblock Weighted ancestors in suffix trees.
\newblock In {\em 22th ESA}, pages 455--466, 2014.

\bibitem{ourICALP}
P.~Gawrychowski, S.~Mozes, and O.~Weimann.
\newblock Improved submatrix maximum queries in {M}onge matrices.
\newblock In {\em 41st ICALP}, pages 525--537, 2014.

\bibitem{KaplanMozesNussbaumSharir}
H.~Kaplan, S.~Mozes, Y.~Nussbaum, and M.~Sharir.
\newblock Submatrix maximum queries in {M}onge matrices and {M}onge partial
  matrices, and their applications.
\newblock In {\em 23rd SODA}, pages 338--355, 2012.

\bibitem{KK89}
M.~M. Klawe and D~J. Kleitman.
\newblock An almost linear time algorithm for generalized matrix searching.
\newblock {\em SIAM Journal Discret. Math.}, 3(1):81--97, 1990.

\bibitem{Kopelowitz}
T.~Kopelowitz and M.~Lewenstein.
\newblock Dynamic weighted ancestors.
\newblock In {\em 18th SODA}, pages 565--574, 2007.

\bibitem{MosheSurvey}
M.~Lewenstein.
\newblock Orthogonal range searching for text indexing.
\newblock In {\em Space-Efficient Data Structures, Streams, and Algorithms -
  Papers in Honor of J. Ian Munro on the Occasion of His 66th Birthday}, volume
  8066, pages 267--302. Springer, 2013.

\bibitem{Nekrich}
Y.~Nekrich.
\newblock Orthogonal range searching in linear and almost-linear space.
\newblock {\em Comput. Geom.}, 42(4):342--351, 2009.

\bibitem{PT2006}
M.~P{\v{a}}tra{\c{s}}cu and M.~Thorup.
\newblock Time-space trade-offs for predecessor search.
\newblock In {\em 38th STOC}, pages 232--240, 2006.

\bibitem{YuanA10}
H.~Yuan and M.~J. Atallah.
\newblock Data structures for range minimum queries in multidimensional arrays.
\newblock In {\em 21st SODA}, pages 150--160, 2010.

\end{thebibliography}

\end{document}